\documentclass[preprint,12pt]{elsarticle}

\usepackage{amssymb}

\usepackage{amsmath,amssymb,amsfonts}
\usepackage{hyperref}
\usepackage[dvipsnames]{xcolor}

\usepackage[colorinlistoftodos,prependcaption,textsize=tiny]{todonotes}
\usepackage{microtype}
\usepackage[utf8]{inputenc}
\usepackage{graphicx}
\usepackage{textcomp}
\usepackage{xcolor}
\usepackage{subfigure}

\usepackage[noend]{algpseudocode}
\usepackage{tabularx}
\usepackage{booktabs}
\usepackage{paralist}
\usepackage{etoolbox}
\usepackage{enumitem}

\usepackage{rotating}
\usepackage{booktabs}
\usepackage{multirow}
\usepackage{verbatim }
\usepackage{booktabs,graphicx}
\usepackage[ruled,linesnumbered]{algorithm2e}

\usepackage{mathptmx}
\usepackage{amsmath}
\usepackage{amssymb}
\usepackage{amsthm}

\newtheorem{theorem}{Theorem}[section]
\newtheorem{lem}[theorem]{Lemma}
\newtheorem{prop}[theorem]{Proposition}

\theoremstyle{definition}
\newtheorem{definition}{Definition}[section]

\newtheorem{defn}{Definition}[section]

\usepackage{booktabs,graphicx}

\usepackage[ruled,linesnumbered]{algorithm2e}
\usepackage{array, boldline, rotating}

\journal{Information Systems}

\begin{document}

\begin{frontmatter}



\title{Differentially Private Release of Event Logs for Process Mining}


\author[inst1]{Gamal Elkoumy\corref{cor1}}
\ead{gamal.elkoumy@ut.ee}
\affiliation[inst1]{organization={University of Tartu},
            city={Tartu},
            country={Estonia}}

\cortext[cor1]{Corresponding author.}
\author[inst2]{Alisa Pankova}
\ead{alisa.pankova@cyber.ee}

\author[inst1]{Marlon Dumas}
\ead{marlon.dumas@ut.ee}

\affiliation[inst2]{organization={Cybernetica},
            city={Tartu},
            country={Estonia}}

\begin{abstract}
The applicability of process mining techniques hinges on the availability of event logs capturing the execution of a business process. 
In some use cases, particularly those involving customer-facing processes, these event logs may contain private information.
Data protection regulations restrict the use of such event logs for analysis purposes. 
One way of circumventing these restrictions is to anonymize the event log to the extent that no individual can be singled out using the anonymized log.
This article addresses the problem of anonymizing an event log in order to guarantee that, upon release of the anonymized log, the probability that an attacker may single out any individual represented in the original log does not increase by more than a threshold.
The article proposes a differentially private release mechanism, which samples the cases in the log and adds noise to the timestamps to the extent required to achieve the above privacy guarantee.
The article reports on an empirical comparison of the proposed approach against the state-of-the-art approaches using 14 real-life event logs in terms of data utility loss and computational efficiency.
\end{abstract}


\end{frontmatter}


\section{Introduction}
\label{sec:intro}

Process mining is a family of techniques that helps organizations enhance their business processes' performance, conformance, and quality. The input of process mining techniques is an event log. An event log captures the execution of a set of instances of a process (herein called \emph{cases}). An event log consists of event records. Each record contains a reference to a case identifier, a reference to an activity, and at least one timestamp. 
Table~\ref{tbl:event_log}
shows an example of a log. 
Each case ID refers to a person (e.g., a patient in a hospital). Each event corresponds to an activity performed for that person. 
For example, in a healthcare process, each activity may correspond to a treatment that the patient in question underwent.


\begin{table}[hbtp]

\label{tbl:event_log}
\centering
\caption{Example of an event log}

\scriptsize
	\begin{tabular}[t]{|c|c|c|c|}
	
\hline

Case ID	&	Activity	&	Timestamp	&	Other Attributes	\\ \hline
\multirow{3}{*}{1}	&	A	&	8/8/2020 10:20	&	….....	\\
	&	B	&	8/8/2020 10:50	&	….....	\\
	&	C	&	8/8/2020 16:15	&	….....	\\ \hline
\multirow{4}{*}{2}	&	D	&	8/8/2020 12:37	&	….....	\\ 
	&	A	&	8/8/2020 14:37	&	….....	\\
	&	E	&	8/8/2020 15:07	&	….....	\\
	&	C	&	8/8/2020 20:31	&	….....	\\\hline
\multirow{3}{*}{3}	&	A	&	8/9/2020 13:30	&	….....	\\
	&	B	&	8/9/2020 13:55	&	….....	\\
	&	C	&	8/9/2020 20:55	&	….....	\\\hline
\multirow{4}{*}{4}	&	D	&	8/9/2020 15:00	&	….....	\\
	&	A	&	8/9/2020 17:00	&	….....	\\
	&	B	&	8/9/2020 17:40	&	….....	\\
	&	C	&	8/9/2020 23:05	&	….....	\\ \hline
\multirow{3}{*}{5}	&	A	&	8/9/2020 17:25	&	….....	\\
	&	E	&	8/9/2020 17:55	&	….....	\\
	&	C	&	8/10/2020 23:55	&	….....	\\ \hline
\multirow{3}{*}{6}	&	A	&	8/11/2020 17:00	&	….....	\\
	&	B	&	8/11/2020 17:27	&	….....	\\
	&	C	&	8/11/2020 23:45	&	….....	\\ \hline

	\end{tabular}

\end{table}

Often, an event log contains private information about individuals. 
Data regulations, e.g., the General Data Protection Regulation (GDPR)\footnote{\url{http://data.europa.eu/eli/reg/2016/679/oj}}, restrict the use of such logs. 
One way to overcome these restrictions is to anonymize the event log such that no individual can be singled out. 
 Singling out an individual happens when they 
 can be distinguished by evaluating a predicate that discriminates them within a group. The legal notion of singling out has been mathematically formalized by Cohen \& Nissim~\cite{cohen2020towards}, who define a type of attack called \emph{Predicate Singling Out (PSO) attack}.
The release of the event log in Table~\ref{tbl:event_log}
permits singling out individuals. Specifically, the predicate ``undergoing activity E after activity D and A'' and the time difference between these activities can lead to a linkage attack~\cite{rafiei2021group}. Given that prior knowledge, the adversary can single out the patient with case ID 2. 

Privacy-Enhancing Technologies (PETs), such as k-anonymity and differential privacy~\cite{dwork2014algorithmic}, 
protect the data release, including event logs.
Among 
existing PETs, differential privacy (DP) stands out because it 
mitigates PSO attacks~\cite{cohen2020towards} and due to its composable privacy guarantees~\cite{dwork2014algorithmic}. 

DP mechanisms 
inject noise into the data, quantified by a parameter called $\epsilon$. 
Lee et al.~\cite{lee2011much} show that ``the proper value of $\epsilon$ varies depending on individual values'' and that the presence of ``outliers also changes the appropriate value of $\epsilon$.'' Dwork et al.~\cite{dwork2019differential} state that ``we do not know what parameter $\epsilon$ is right for any given differentially
private analysis, and we do know that the answer can vary tremendously 
based on attributes of the dataset.'' Accordingly, this article proposes a method to 
determine the $\epsilon$ value required to anonymize an event log based on a business-level metric, namely \emph{guessing advantage}. The guessing advantage is the 
increase in the probability that an adversary may guess information 
about an individual after the event log's release. 

Usually, an adversary (e.g., an analyst) has prior knowledge about individuals in the log before its release. For example, the attacker may know that the person went to the hospital at 10 am Sunday morning and they received a rare treatment. Using this knowledge, the adversary has a certain probability of guessing  information about an individual. After the release, the adversary gains an additional advantage (knowledge) to guess information. Anonymization limits this risk. 
This article investigates limiting the additional risk by a maximum guessing advantage level $\delta$.
Specifically, the following problem is addressed:

\begin{quote} \it
Given an event log L, and given a maximum level of acceptable guessing advantage $\delta$, generate an anonymized event log L$'$ such that the success probability of singling out an individual after publishing L$'$ does not increase by more than $\delta$.
\end{quote}

Naturally, we should ensure that the anonymized log is useful for process mining. 
In this respect, a desirable property, at least in some use cases, is that the anonymized log should not introduce new \emph{case variants} relative to the original one. A case variant is a distinct sequence of activities. For example, the case variants of the log in Table~\ref{tbl:event_log} are $\{ \langle A, B, C \rangle , \langle D,A, E, C \rangle , $ $  \langle D, A, B, C \rangle ,   \langle A, E, C \rangle \}$. 
There are at least two reasons why this property may be desirable. First, this property ensures that the anonymization does not introduce spurious directly-follows relations between activities. The set of directly-follows relations between activities, known as the Directly-Follows Graph (DFG), is a common artifact produced by process mining tools to help their users to understand the structure of a process~\cite{LeemansPW19,Chapela-CampaDM22}. The DFG is also used as input by many automated process discovery techniques~\cite{augustoCDRMMMS19}. Spurious directly-follows relations may lead users into wrong conclusions, for example, a spurious relation in a DFG may lead the user to conclude that activity ``Pay invoice'' is sometimes followed by ``Receive invoice'', even when this never happens in the original log.
Second, this property ensures that a conformance checking algorithm~\cite{Carmona22} applied to the anonymized log does not return false positives, i.e., that it does not report deviations that do not exist in the original log. Indeed, every case variant added to an event log during anonymization is a potential false positive in the conformance checking output. Notwithstanding the above, we acknowledge that there are use cases, e.g.\ in the field of process performance mining~\cite{kabierski2023hiding}, where this property might not be required.


A second desirable property is that the differences between the timestamps of consecutive events in the anonymized log are as close as possible to those in the original log, as these time differences are used by performance mining techniques~\cite{DumasRMR18}.
Accordingly, we tackle the above problem subject to the requirements:

%
\begin{enumerate}[label=\textbf{R\arabic*}]

	\item\label{int:req:trace}
	The anonymized event log must not introduce new case variants to the original log. 
	
	
	 
	\item\label{int:req:time}
	The difference between the real and the anonymized time values is minimal given the risk metric $\delta$. 

\end{enumerate}
The second requirement can be tackled w.r.t. different attack models. This article considers an attack model wherein the attacker seeks to single out an individual, represented by a trace in the log, based on a prefix, a suffix, or an event timestamp of the individual's trace in the released log.


We tackle the above problem by proposing a notion of \emph{differentially private event log}. Given a maximum allowed guessing advantage, $\delta$, a differentially private event log is obtained by sampling the traces in the log and injecting noise to the event timestamps. After release of a differentially private event log, the probability that an attacker may single out any individual, based on the prefixes, suffixes, or event timestamps of the individual's trace in the released log, is not more than $\delta$.





This article is an extended version of a conference paper~\cite{elkoumy2021mine}. Relative to the conference version, the main extension is a revised anonymization method, which achieves lower utility loss for a given guessing advantage $\delta$ by: (i) applying both over- and undersampling, as opposed to only oversampling; and (ii) filtering out high-risk cases. The article also extends the evaluation to assess utility loss 
w.r.t.\ the impact of anonymization on the process maps discovered from the 
log.


The article is structured as follows. Sect.~\ref{sec:background} introduces background notions and related work. 
Sect.~\ref{sec:approach} presents the proposed approach. 
Sect.~\ref{sec:eval} presents an empirical evaluation. Finally, Sect.~\ref{sec:conclusion} concludes and discusses future work.

\section{Background and Related Work}
\label{sec:background}
In this section, we introduce bounded and unbounded DP, and we formalize their definitions. We then overview the recently developed privacy-preserving process mining techniques.

\subsection{Differential Privacy (DP)}

\label{sec:diffPriv}



We consider the problem of providing differential privacy guarantees on the release of an event log.
As illustrated in Table~\ref{tbl:event_log}, an event log is a set of events capturing the execution of activities of a process. Each event contains a unique identifier of the process instance in which it occurs (the case ID), an activity label, and a timestamp. 
An event may contain other attributes, e.g., resources. 
This article focuses on anonymizing three attributes: case ID, activity label, and timestamp.
If we group the events in a log by case ID and sort the events chronologically, every resulting group is called a trace. A trace captures the sequence of events that occurred in a case that  corresponds to an individual (e.g.,  a customer) who requires their privacy to be maintained. If an attacker can single out a trace, they can single out the corresponding individual.
\begin{defn} [Event Log, Event, Trace]\label{def:event_log}
	An event log $L= \{e_1, e_2, ..., e_n\}$ of a process is a set of events $e=(i,a,ts)$, each capturing an execution of an activity $a$ (an activity instance), with a timestamp $ts$, as part of a case $i$ of the process. The trace $t=\langle e_1, e_2, ..., e_m\rangle$ of a case $i$ is the sequence of events in $L$ with identifier $i$, ordered by timestamp. An event log $L$ may be represented as a set of traces $\{ t_1, t_2, ..., t_k \}$. 
\end{defn}


This paper proposes a differentially-private mechanism to anonymize the event log.
A privacy mechanism $M: L \rightarrow Range (M)$ maps an event log $L$ to a particular distribution of values $Range(M)$ (e.g., to a vector of real numbers). A mechanism $M$ can be either unbounded or bounded $\epsilon$-differentially private ($\epsilon$-DP). An unbounded $\epsilon$-DP mechanism makes it hard to distinguish two event logs that differ in the \emph{presence} of one trace~\cite{dwork2006}. A bounded $\epsilon$-DP mechanism makes it hard to distinguish two event logs that differ in the \emph{attribute values} of one trace.

\begin{defn} [Unbounded $\epsilon$-differentially private mechanism~\cite{dwork2006}]\label{def:udp}
A mechanism $M$ is said to be $\epsilon$-differentially private if, for all the event logs $L_1$ and $L_2$ differing \underline{at most on one trace}, and all $S \subseteq Range (M)$, we have $Pr[M(L_1) \in S] \leq exp(\epsilon) \times Pr[M(L_2) \in S]$.
\end{defn}

\begin{defn} [Bounded $\epsilon$-differentially private mechanism~\cite{dwork2006calibrating}]\label{def:bdp}
A mechanism $M$ is $\epsilon$-differentially private if, for all the event logs $L_1$ and $L_2$ differing \underline{at most on the attribute values of one trace}, and all $S \subseteq Range (M)$, we have $Pr[M(L_1) \in S] \leq exp(\epsilon) \times Pr[M(L_2) \in S]$.
\end{defn}

In some cases, it is desired to apply DP to only values of a particular attribute $A$, e.g., the attribute timestamp in Table~\ref{tbl:event_log}.
We apply DP to $L_1$ and $L_2$ w.r.t the timestamp attribute $TS$, i.e., $L_1$ and $L_2$ differ only on $TS$'s value in a single trace.
Moreover, we want to take into account the particular \emph{amount of change} in the attribute $TS$.


\begin{defn} [Bounded $\epsilon$-differentially private mechanism w.r.t the timestamp attribute]\label{def:bdp_attribute}
A mechanism $M$ is $\epsilon$-differentially private w.r.t the timestamp attribute $TS$ iff for every pair of event logs $L_1$ and $L_2$ differing along attribute $TS$ in at most one trace, and for all $S \subseteq Range (M)$, we have 

$Pr[M(L_1) \in S] \leq exp(\epsilon \cdot |L_1.TS - L_2.TS|) \times Pr[M(L_2) \in S]$.

\end{defn}

The $\epsilon$-differential privacy restricts the ability to single out an individual (Def.~\ref{def:udp} and \ref{def:bdp}) or disclose an individual's private attribute (Def.~\ref{def:bdp_attribute}).In an interactive mechanism~\cite{dwork2014algorithmic}, a user submits a query function $f$ to an event log and receives a noisified result. Formally, there is a mechanism  $M_f$ that computes $f$ and injects noise into the result. The amount of 
noise depends on the \textit{sensitivity} of $f$, which quantifies how much change in the input of $f$ affects change in its output.

\begin{definition} [Global Sensitivity] \label{def:gs2}
Let $f : L \rightarrow \mathbb{R}^d$.
\begin{itemize}
\item Global sensitivity w.r.t. presence of a trace is $\Delta f= \max\limits_{L_1,L_2} |f(L_1) - f(L_2)|$;
\item Global sensitivity w.r.t. the timestamp attribute $TS$ is $\Delta^A f= \max\limits_{L_1,L_2} \frac{|f(L_1) - f(L_2)|}{|L_1.TS - L_2.TS|}$;
\end{itemize}
where $\max$ is computed over all event logs $L_1,L_2 $ differing in one trace at most.
\end{definition}

Given the event log $L$ and the query function $f$, a randomized mechanism $M_f$ returns a noisified output $f(L)+Y$, where $Y$ is a noise value drawn randomly from a particular distribution. E.g., we can draw values from a Laplace distribution $Lap(\lambda, \mu)$, which has a probability density function $\frac{1}{2\lambda} exp(-\frac{|x-\mu|}{\lambda})$, where $\lambda$ is a scale factor, and $\mu$ is the mean. It is known~\cite{dwork2014algorithmic} that, for real-valued $f$, if we set $\mu=0$, for $\lambda = \frac{\Delta f}{\epsilon}$ we obtain an $\epsilon$-DP mechanism w.r.t. a trace presence (Def.~\ref{def:udp}), and for $\lambda = \frac{\Delta^A f}{\epsilon}$, we obtain an $\epsilon$-DP mechanism w.r.t. attribute $A$ (Def.~\ref{def:bdp_attribute}).

 The privacy parameter 
 $\epsilon$ ranges from $0$ to $\infty$, and the desired level of privacy 
 depends on the data distribution. 
 Lee et al.~\cite{lee2011much} demonstrate the challenges with choosing $\epsilon$ to protect individual information with a fixed probability. Although  $\epsilon$ is used to quantify the risk of releasing a statistical analysis of sensitive data, it is not an absolute privacy metric but rather a relative value. 
 Hsu et al. \cite{hsu2014differential} present an economical method for choosing $\epsilon$. 
 They consider two conflicting goals: learning the correct analysis from the data and keeping the data of individuals private. 
 They use a privacy budget for individuals to balance the conflicting objectives.
 
 Laud et al.~\cite{laud2020framework} propose a framework to quantify $\epsilon$ from a probability score called the \textit{guessing advantage} -- the increase in the probability that an adversary may guess information about an individual after data release. 
 They state the attacker's goal as a Boolean expression of guessing attributes, and they studied the change of prior and posterior probabilities of guessing. 
 This article  adopts the work in~\cite{laud2020framework} to provide an estimation 
  of $\epsilon$ based on the guessing advantage threshold, with the assumption of the worst case scenario in which an adversary has background knowledge about all the other instances.

Li et al.~\cite{li2012sampling} prove that adding random sampling to a differentially private mechanism amplifies the level of privacy protection. They explain that adding random sampling with a probability $\beta$ reduces the $e^\epsilon$ by a factor of $\beta$. This privacy amplification is valid with ($\epsilon$, $\delta$)-differential privacy under sampling in the setting where one publishes an anonymized version of a dataset~\cite{li2012sampling}. In this article, we perform a sampling of cases based on prefixes/suffixes groups. We sample an entire case in order to fulfill the requirement~\ref{int:req:trace}. Also, we prune cases that are too sensitive and violate the guessing advantage threshold $\delta$. This case pruning reduces the amount of noise, as we will explain later.

\subsection{Privacy-Preserving Process Mining (PPPM)}

\label{sec:pppm}
The use of PETs for PPPM has been considered in previous studies. Elkoumy et al.~\cite{elkoumy2021privacy} 
present three families of privacy models to achieve PPPM: group-based models, DP models, and encryption-based models.
They then analyze the privacy requirements that GDPR brings to PPPM. This article addresses two of these requirements: anonymity and unlinkability~\cite{elkoumy2021privacy}.

Group-based models have been addressed in several studies~\cite{rafiei2020tlkc,rafiei2021group}. TLKC~\cite{rafiei2020tlkc} proposes a k-anonymity mechanism to anonymize event logs. TLKC anonymizes 
a log from the cases' perspective. This work has been extended~\cite{rafiei2021group} to anonymize the log from the resource perspective. These techniques adopt k-anonymity, which results in suppressing cases or events within cases.
In one example in~\cite{rafiei2020tlkc}, TLKC suppressed 87\% of the activities in the output.
Besides, k-anonymity does not fully mitigate PSO attacks~\cite{cohen2020towards}.

Other studies adopt differential privacy. Mannhardt et al.~\cite{mannhardt2019privacy}  propose an event log summarization approach that uses differential privacy to anonymize two types of queries: the query ``frequencies of directly-follows relations'' and ``frequencies of trace variants.'' This approach does not consider logs with timestamps. While summarization approaches are beneficial for some process mining techniques, e.g., constructing the directly-follows graph and trace variant analysis, other process mining techniques such as conformance checking, business process simulation, and performance mining require the release of an anonymized event log where the events annotated with their execution timestamps.


PRIPEL~\cite{fahrenkrog2020pripel} uses timestamp shifts to anonymize the timestamp attribute of the log. It ensures privacy guarantees based on individual cases. Also, PRIPEL uses sequence enrichment to anonymize other attributes of the log. PRIPEL takes three input parameters, namely $\epsilon$, $k$, and $N$. $\epsilon$ is the DP parameter, $k$ is the cut-out frequency (i.e., PRIPEL cuts out all variants that appear less than $k$), and $N$ is the maximal prefix length.
This study does not address the problem stated in Sect.~\ref{sec:intro} because they do not limit the guessing advantage to a certain threshold. Instead, the user has to provide an $\epsilon$ value as an input. PRIPEL~\cite{fahrenkrog2020pripel} uses the same $\epsilon$ value to anonymize the trace variants and timestamps, although $\epsilon$ has a different interpretation for different value ranges. 

SaCoFa~\cite{fahrenkrog2021sacofa} uses semantic constraints to achieve high-utility anonymization of event logs using DP. This approach anonymizes the log by replacing prefixes shared by multiple case variants with perturbed versions thereof. The approach replaces an original prefix with a perturbed one, provided that the perturbed prefix fulfills certain semantic constraints and is within a certain distance of the original one. SaCoFa focuses on anonymizing the control flow query  and does not address the problem stated in Sect.~\ref{sec:intro} insofar as it does not anonymize the timestamps.




Other studies anonymize the event log from the resource perspective.  Batista \& Solanas~\cite{batista2021uniformization} present an approach that 
groups individuals based on activity distribution and transfers resource information within the group to uniform the resource distribution. PRETSA~\cite{fahrenkrog2019pretsa} provides event log sanitization by adopting k-anonymity and t-closeness to avoid disclosing employees' identities. These studies do not address the problem stated in Sect.~\ref{sec:intro} because they study log anonymization to protect resource disclosure.


Other related work of 
PPPM is orthogonal to our research, which does not provide a concrete method to anonymize the event log. 
Kabierski et al.~\cite{kabierski2021privacy,kabierski2023hiding} provide a framework for anonymizing process performance indicators.
Rafiei et al.~\cite{rafiei2020towards} present quantification metrics of both utility and disclosure risk. 
Furthermore, other studies addressed the 
secure processing of distributed event logs~\cite{elkoumy2020secure,elkoumy2020shareprom}. 
Other studies addressed the privacy-preserving continuous event data publishing~\cite{rafiei2021privacy,Rafiei22Quantifying}.


\section{Approach}
\label{sec:approach}


We seek to anonymize an event log in such a way that an attacker cannot single out an individual based on a prefix or suffix of the individual's trace or based on the event timestamps. Accordingly, the proposed approach relies on a data structure that captures all prefixes and suffixes of a set of traces, namely a Deterministic Acyclic Finite State Automata (DAFSA)~\cite{reissner2017scalable}. 
By analyzing the frequency and time differences of each DAFSA transition, we estimate the amount of noise required to achieve the required guessing advantage.

Concretely, given a log and a guessing advantage threshold $\delta$, our approach produces a differentially private event log in 7 steps as outlined in Fig.~\ref{fig:approach}. 
\begin{figure}[hbtp]
\centering
\includegraphics[width=1\columnwidth]{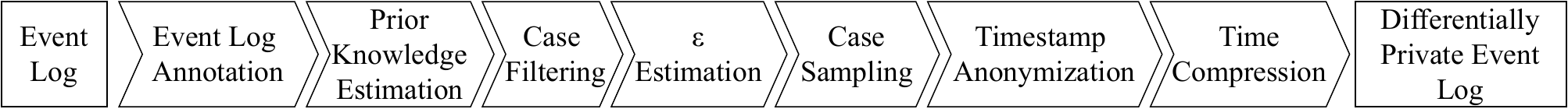}

\caption{ Approach }
	\label{fig:approach}
   
\end{figure}
The first step constructs a lossless intermediate representation of a log (DAFSA). The DAFSA groups the traces that share the same prefixes or suffixes. An attacker may have prior knowledge about the events recorded in the log before the log publishing. Consequently, we estimate the prior knowledge of every event that belongs to a prefix/suffix group. Some events may correlate with high prior-knowledge values, leading to more noise injected in the log to achieve the given $\delta$ threshold. Accordingly, we provide a case filtering mechanism that filters out entire cases based on the estimated prior knowledge of their events. The fourth step  estimates an $\epsilon$ value for every prefix/suffix group for the input $\delta$ threshold. The fifth step uses the estimated $\epsilon$ value to apply sampling to the cases based on their prefixes/suffixes and case variants. The sixth step applies timestamp anonymization based on the estimated $\epsilon$ values for every group of prefixes/suffixes. Lastly, we post-process the differentially private log to compress the timestamp values so that the overall timeframe of the resulting log matches closely with the original log. We generate new case IDs, so an attacker cannot use the case ID (on its own) to identify an individual.
The rest of this section introduces the attack model and then discusses each step of the proposed approach in turn.


\subsection{Attack Model}
\label{sec:attack_model}
We consider a scenario where a data owner shares an event log with an analyst.
We assume that the activity labels and the smallest and largest timestamp in the event log are public information. We assume that each trace in the log pertains to an individual whose privacy we wish to safeguard. We view the analyst as a potential attacker who may seek to infer information about an individual based on the released log. We seek to protect the release under a worst-case scenario where the analyst has background knowledge about all individuals in the log except for the individual of their interest.
Specifically, we seek to protect the release of log L to prevent the attacker from fulfilling one or both of the two attack goals:
\begin{itemize}
    \item $h_1$: Has the case of an individual gone through a given prefix or suffix? The output is a bit with a value $\in \{0,1\}$.
    \item $h_2$: What is the cycle time of a particular activity that has been executed for the individual? The output is a real value that the attacker may wish to estimate with a certain precision.
\end{itemize}


Note that we do not seek to prevent the attacker from guessing the activity labels, i.e.,\ we do not view the activity labels as private information. Also, we assume that cases are independent, meaning that the sequence of activities that a case follows does not depend on the activity sequences of other cases. This assumption usually holds in a business process, e.g.,\ the patient's pathway in a treatment process does not depend on that of other patients. If an individual participates in more than one case (i.e., cases are not independent), the $\epsilon$ parameter should be divided by the maximum number of cases related to an individual to estimate the decreased privacy guarantees. 
Also, this paper proposes a one-shot event log release mechanism where the
log owner anonymizes the log once and releases the anonymized log as the single access to their business process execution. In case of repeating the anonymization more than once, the $\epsilon$ parameter shall be divided by the number of repeated anonymizations.

To prevent an attacker from achieving goals $h_1$ and $h_2$, we introduce a notion of a differentially private event log. 

\begin{defn}[Differentially-Private Event Log]
\label{def:diff_priv_EL}
Let $L$ be an event log as defined in Def.~\ref{def:event_log}. We say that a log $L' = M(L)$ is $\epsilon$-differentially private if: (1) it ensures $\epsilon$-differential privacy guarantee from the control-flow perspective; (2) it is $\epsilon$-differentially private w.r.t. timestamp.
\end{defn}

\subsection{Event Log State-Annotation}

\label{sec:app:prep}
\label{sec:event_log_rep}



Our goal is to prevent an attacker from singling out individuals based on any prefix or suffix of their trace (cf.\ attack goal $h_1$). To this end, we group the prefixes and suffixes in the log and inject independent differentially private noise to each group. Consequently, we need a log representation that partitions the prefixes and suffixes of the log traces into groups. In other words, this representation should assign each prefix (suffix) in the log to one group such that the union of the groups is equal to the entire set of prefixes (suffixes). Also, we require that this representation preserves only the set of case variants of the log (cf.~\ref{int:req:trace}).

The DAFSA provides us with such a partitioning. Given a set of words, each state in the DAFSA corresponds to a group of prefixes that share the same set of possible suffixes and suffixes that share the same set of prefixes~\cite{daciuk2000incremental}. 
An advantage of the DAFSA (specifically the \emph{minimal DAFSA}) over similar representations, such as prefix trees, is that a (minimal) DAFSA contains a minimal number of groups (states)~\cite{daciuk2000incremental}. By minimizing the number of groups, we obtain larger groups. The larger the group, the smaller the needed noise injection to achieve $\epsilon$-DP.


\begin{defn}[Minimal DAFSA of a set of words~\cite{daciuk2000incremental}]\label{def:dafsa}
Let $V$ be a finite set of labels. A DAFSA is an acyclic directed graph $D = ( S, s_0, A, S_f)$, where $S$ is a finite set of states, $s_0 \in S$ is the initial state, $A \subset S \times V \times S$ is a set of labeled transitions, and $S_f$ is a set of final states. A DAFSA of a set of words $W$ is a DAFSA such that every word in $W$ is a path from an initial to a final state, and, conversely, every path from an initial state to a final state is a word in $W$. A minimal DAFSA of a set of words $W$ is a DAFSA of $W$ with a minimal number of states. 
\end{defn}

Given a path from the initial state $s_0$ to a state $s\in S$ in a DAFSA, the sequence of labels of the arcs in the path is the \textit{prefix} of $s$. Similarly, given a path from $s$ to a final state $s_f \in S_f$, the sequence of labels 
in such  path is the \textit{suffix} of $s$. 
 

Reissner et al.~\cite{reissner2017scalable} reuse the algorithm in~\cite{daciuk2000incremental} to represent an event log as a DAFSA. Every trace in the log is seen as a word over the alphabet of activity labels. Each particular word extracted from an event log in this way is called a \emph{case variant} of the log.

\begin{defn}[Case Variant] 
Given an event log $L$, a case variant of $L$ is a sequence of activity labels $\langle\, a_1, a_2, ..., a_k \,\rangle$ such that there is a trace $\langle\, e_1, e_2, ..., e_k \,\rangle$ of $L$ such that $\forall j \in [1..k] \; id(e_j) = a_j$, where $id(e)$ is the case ID of the event $e$.
\end{defn}

The set of case variants (each one represented as a word) is then compressed into a DAFSA. 
For example, the DAFSA representation of the log in Table~\ref{tbl:event_log} is shown in Fig.~\ref{fig:dafsa}. 




\begin{figure}[htbp]
  
  \centering
\includegraphics[width=.7\columnwidth]{{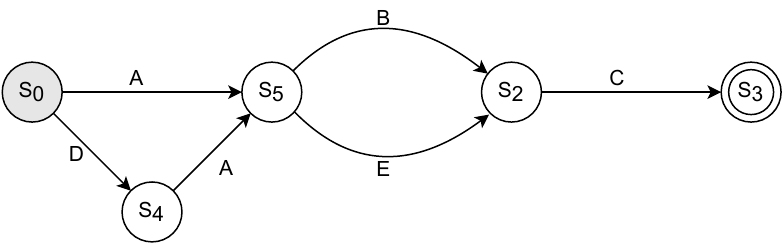}}

	\caption{ DAFSA of the event log in Table~\ref{tbl:event_log}}
	\label{fig:dafsa}
  
\end{figure}

 \begin{defn}[Common prefixes and suffixes~\cite{reissner2017scalable}] \label{def:common}
 Let $D = ( S, s_0, A, S_f)$ be a DAFSA. The set of common prefixes of $D$ is $\mathbb{P} =\{pref(s) | s\in S \wedge |s\blacktriangleright|>1\}$. The set of common suffixes of $D$ is $\mathbb{S}= \{suff(s) | s\in S \wedge |\blacktriangleright s|>1\}$.
 \end{defn}

The common prefixes of the DAFSA in Fig.~\ref{fig:dafsa} are $\{\langle A,B\rangle\, \langle D,A\rangle\, \langle A\rangle \}$, and the common suffixes are $\{\langle B,C\rangle\, \langle E,C\rangle \}$.
Cases corresponding to case variants that traverse a given DAFSA state $s$ share the same 
set of prefixes and suffixes. This article employs DAFSA states and transitions 
to group common 
prefixes and suffixes within cases of the log. We annotate the log in such a way that each event (activity instance) is related to a DAFSA transition (i.e.,\ each event is annotated with a pair of states).

\begin{defn} [DAFSA-Annotated Event Log]
\label{def:state_el}
A DAFSA-annotated event log $L_s =\{r_1, r_2, ..., r_n\}$ is a set of entries $r=(i,a,ts, s_i, s_e)$, each links an event $e=(i,a,ts) \in L$, where $L$ is the event log, to the DAFSA transition $t=(s_i, a, s_e)$ that represents the occurrence of that event, where $t$ starts from $s_i$, ends at $s_e$, and labeled with the same activity $a$.

\end{defn}

The DAFSA-annotated event log links every event in the log, based on its prefix, suffix, and activity label, to a DAFSA transition. 
Every event is labeled by the source state and target state of the DAFSA transition. 
Table~\ref{tbl:state_annotated} (columns ``Src. state'' and ``Tgt. state'') shows the DAFSA-annotated log.\footnote{Columns ``Norm. Rel. Time'', ``Prec.'' and ``PK'' are explained later.}

\begin{table}[hbtp]

	\centering
\caption{DAFSA State-Annotated Event log}

\scriptsize
	\begin{tabular}[t]{|p{0.4cm}|p{0.4cm}|c|p{0.4cm}|p{0.4cm}|p{0.4cm}|p{0.5cm}|p{0.4cm}|p{0.4cm}|}
\hline

\shortstack{Case\\ID}&\shortstack{Act.\\Label}&Timestamp&\shortstack{Src.\\State}&\shortstack{Tgt\\State}&\shortstack{Rel.\\Time}&	\textit{\shortstack{Nrm.\\Rel.\\Time}}	&	\textit{Prec.}	&	\textit{PK}	\\ \hline
\multirow{3}{*}{1}	&	A	&	8/8/2020 10:20	&	$s_0$	&	$s_5$	&	0	&	0	&	1	&	0.33	\\
	&	B	&	8/8/2020 10:50	&	$s_5$	&	$s_2$	&	30	&	0.33	&	0.67	&	0.75	\\
	&	C	&	8/8/2020 16:15	&	$s_2$	&	$s_3$	&	325	&	0.0	&	0.01	&	0.33	\\ \hline

\multirow{4}{*}{2}	&	D	&	8/8/2020 12:37	&	$s_0$	&	$s_4$	&	0	&	0	&	1	&	0.33	\\
	&	A	&	8/8/2020 14:37	&	$s_4$	&	$s_5$	&	120	&	1	&	0	&	0.35	\\
	&	E	&	8/8/2020 15:07	&	$s_5$	&	$s_2$	&	30	&	1	&	0	&	0.35	\\
	&	C	&	8/8/2020 20:31	&	$s_2$	&	$s_3$	&	324	&	0	&	0.01	&	0.33	\\
 \hline

\multirow{3}{*}{3}	&	A	&	8/9/2020 13:30	&	$s_0$	&	$s_5$	&	0.99	&	0.33	&	1	&	0.5	\\
	&	B	&	8/9/2020 13:55	&	$s_5$	&	$s_2$	&	25	&	0	&	0.67	&	0.5	\\
	&	C	&	8/9/2020 20:55	&	$s_2$	&	$s_3$	&	420	&	0.07	&	0.01	&	0.167	\\
\hline

\multirow{4}{*}{4}	&	D	&	8/9/2020 15:00	&	$s_0$	&	$s_4$	&	0.99	&	1	&	1	&	0.5	\\
	&	A	&	8/9/2020 17:00	&	$s_4$	&	$s_5$	&	120	&	1	&	0	&	0.35	\\
	&	B	&	8/9/2020 17:40	&	$s_5$	&	$s_2$	&	40	&	1	&	0.67	&	0.25	\\
	&	C	&	8/9/2020 23:05	&	$s_2$	&	$s_3$	&	325	&	0.0	&	0.01	&	0.33	\\ \hline

\multirow{3}{*}{5}	&	A	&	8/9/2020 17:25	&	$s_0$	&	$s_5$	&	0.99	&	0.33	&	1	&	0.5	\\
	&	E	&	8/9/2020 17:55	&	$s_5$	&	$s_2$	&	30	&	1	&	0	&	0.35	\\
	&	C	&	8/10/2020 23:55	&	$s_2$	&	$s_3$	&	1800	&	1	&	0.01	&	0.17	\\ \hline

\multirow{3}{*}{6}	&	A	&	8/11/2020 17:00	&	$s_0$	&	$s_5$	&	3	&	1	&	1	&	0.17	\\
	&	B	&	8/11/2020 17:27	&	$s_5$	&	$s_2$	&	27	&	0.13	&	0.67	&	0.5	\\
	&	C	&	8/11/2020 23:45	&	$s_2$	&	$s_3$	&	378	&	0.04	&	0.01	&	0.17	\\
 \hline

	\end{tabular}
\label{tbl:state_annotated}

\end{table}

Timestamps represent the time at which every event happened. However, it is more beneficial to compare the duration of activities in every group of suffixes and prefixes to prevent singling out an individual based on the duration of their activity and its timestamp. We calculate the relative time to compare the time within a group of traces that share the same suffix/prefix and activity label. The relative time of an event is the time difference between an event and its successor. Table~\ref{tbl:state_annotated} (column Rel. Time) shows the relative time estimated for the given events. For the start event of every case, the relative time is the difference between the case start event and the first event in the event log.

\subsection{Prior Knowledge Estimation}

\label{sec:app:prior}
Before the log release, attackers can use their knowledge to guess information about a specific individual. 
We estimate the prior knowledge of an attacker using the framework proposed in~\cite{laud2020framework}.
An attacker's guess $h(L)$ is considered successful if it falls within a range of values $H_p$, which is the actual value $\pm$ precision.

\begin{defn}[Prior Guessing Probability]\label{def:prior}
An attacker's prior guessing probability is $P := Pr[h(L) \in H_p]$


\end{defn}




A guessing precision $p$ is a percentage value representing the range of a successful guess $H_p$. For example, if the true value is 0.5 and $p= 0.2$, the guessed value is considered 
successful if it falls in range $H_p= [0.3,0.7]$. To interpret the precision in the range of values [0,1] (a percentage value), we pre-process the log to normalize the range of values (relative timestamps) to be in the range [0,1]. 
We assume that the precision for the start timestamp is one day and the precision for the relative time is 10 seconds.
Table~\ref{tbl:state_annotated} (columns Nrm. Rel. Time and Prec.) shows the normalized relative time (the normalization is based on the event's DAFSA transition group) and the estimated precision.

The prior knowledge $P_k$ can be estimated based on the data distribution as~\cite{laud2020framework}:

%
\begin{equation}\label{eqn:p_k}
\footnotesize
P_{k} = CDF(t_{k} + p \cdot r) - CDF(t_{k} - p \cdot r)\enspace,
\end{equation}

\noindent where CDF is the cumulative distribution function, $t_k$ is the value of 
an instance $k$ that falls in a range of values, $p$ is the  precision, and $r$ is the upper bound of the range of values.

Suppose we cannot estimate the probability distribution of input values due to the absence of data distribution (e.g., the likelihood of participating in a subtrace). In that case, we can compute the worst-case scenario $P_k$ as~\cite{laud2020framework}:
%
\begin{equation}\label{eqn:p_freq}
\footnotesize
P_k=(1-\delta) / 2  \text{ for all $k$}\enspace.
%
\end{equation}

To apply the prior knowledge to the DAFSA-annotated log, we group the events based on their DAFSA transition (source state, activity, and target state). For each group of values, we estimate the prior knowledge using Eq.~\ref{eqn:p_k}. Table~\ref{tbl:state_annotated} (column PK) shows the estimated prior knowledge for every instance with $\delta =0.3$.




\subsection{Case Filtering}

\label{sec:app:filtering}
In some cases, the prior knowledge $P_k$ is very high such that more noise is needed to keep the guessing advantage below the threshold $\delta$ after publishing the event log. We filter out the instances that violate $P_k+\delta \geq 1$ to reduce the injected noise in such a case. To fulfill~\ref{int:req:trace}, we filter out the entire case.

\begin{defn} [Case Filtering]
\label{def:case_filtering}
A case filter F is a function that filters out cases with at least one DAFSA transition that violates the condition $P_k+\delta \geq 1$.
\end{defn}

In Table~\ref{tbl:state_annotated}, activity B of the first case has a prior knowledge value $P_k=0.75$ and $\delta=0.3$. Our approach filters out the entire first case because at least one DAFSA transition violates the condition $P_k+\delta \geq 1$. Table~\ref{tbl:filtered_log} shows the filtered DAFSA-annotated event log.\footnote{The column ``$\epsilon_t$'' is explained later.} After case filtering, we re-estimate the prior knowledge for each event. Table~\ref{tbl:filtered_log} (column New PK) presents the newly estimated prior knowledge values.  Case filtering does not have an impact on privacy because it does not depend on private data.




\begin{table}[hbtp]
	
\caption{Filtered DAFSA State-Annotated Event log}

\scriptsize
\centering
\begin{tabular}[t]{|c|c|c|c|c|c|c|c|}
\hline

\shortstack{Case\\ID}&\shortstack{Act.\\Label}&Timestamp&\shortstack{Src\\State}&\shortstack{Tgt\\State}&\shortstack{Rel.\\Time}&\shortstack{New\\PK}&\textit{$\epsilon_t$}\\ \hline

\multirow{4}{*}{2}	&	D	&	8/8/2020 12:37	&	$s_0$	&	$s_4$	&	0	&	0.2	&	1.39	\\
	&	A	&	8/8/2020 14:37	&	$s_4$	&	$s_5$	&	120	&	0.35	&	1.24	\\
	&	E	&	8/8/2020 15:07	&	$s_5$	&	$s_2$	&	30	&	0.35	&	1.24	\\
	&	C	&	8/8/2020 20:31	&	$s_2$	&	$s_3$	&	324	&	0.2	&	1.39 \\
 \hline

\multirow{3}{*}{3}	&	A	&	8/9/2020 13:30	&	$s_0$	&	$s_5$	&	0.99	&	0.6	&	1.8	\\
	&	B	&	8/9/2020 13:55	&	$s_5$	&	$s_2$	&	25	&	0.33	&	1.24	\\
	&	C	&	8/9/2020 20:55	&	$s_2$	&	$s_3$	&	420	&	0.2	&	1.39	\\ \hline

\multirow{4}{*}{4}	&	D	&	8/9/2020 15:00	&	$s_0$	&	$s_4$	&	0.99	&	0.6	&	1.79	\\
	&	A	&	8/9/2020 17:00	&	$s_4$	&	$s_5$	&	120	&	0.35	&	1.24	\\
	&	B	&	8/9/2020 17:40	&	$s_5$	&	$s_2$	&	40	&	0.33	&	1.24	\\
	&	C	&	8/9/2020 23:05	&	$s_2$	&	$s_3$	&	325	&	0.2	&	1.39	\\ \hline

\multirow{3}{*}{5}	&	A	&	8/9/2020 17:25	&	$s_0$	&	$s_5$	&	0.99	&	0.6	&	1.79	\\
	&	E	&	8/9/2020 17:55	&	$s_5$	&	$s_2$	&	30	&	0.35	&	1.24	\\
	&	C	&	8/10/2020 23:55	&	$s_2$	&	$s_3$	&	1800	&	0.2	&	1.39	\\ \hline

\multirow{3}{*}{6}	&	A	&	8/11/2020 17:00	&	$s_0$	&	$s_5$	&	3	&	0.2	&	1.39	\\
	&	B	&	8/11/2020 17:27	&	$s_5$	&	$s_2$	&	27	&	0.33	&	1.24	\\
	&	C	&	8/11/2020 23:45	&	$s_2$	&	$s_3$	&	378	&	0.2	&	1.39	\\
\hline

\end{tabular}
\label{tbl:filtered_log}

\end{table}

\subsection{ \texorpdfstring{$\mathcal\epsilon$}{Lg} Estimation}

\label{sec:app:eps}
Given the filtered DAFSA-annotated event log, we need to estimate the amount of noise in order to anonymize the event log. We use DP, which quantifies the noise using $\epsilon$. One naive approach is using the same $\epsilon$ both for control-flow and timestamp anonymization. The two types of queries are different, and the quantity of noise has a different impact on each.
Given the attack model in Sect.~\ref{sec:attack_model}, 
for attack $h_1$, we use the lossless log representation (DAFSA) 
to apply an $\epsilon$-DP mechanism to the prefixes and suffixes of the traces in the log. For attack $h_2$, we apply a bounded $\epsilon$-DP mechanism w.r.t. the timestamp attribute (Def.~\ref{def:bdp_attribute}).
Consequently, we seek to quantify two $\epsilon$ values: the $\epsilon_d$ is required to anonymize the DAFSA, and the $\epsilon_t$ is required to anonymize the timestamp attribute.

Both $\epsilon_d$ and $\epsilon_t$ should guarantee that the attacker cannot single out an individual based on any prefix or suffix of their trace. Accordingly, we group the prefixes and suffixes in the log and estimate $\epsilon_d$ and $\epsilon_t$ for each group. 

\subsubsection{Estimating $\epsilon_t$}

To protect against attack goal $h_2$, we need to inject noise sufficient to achieve a given $\epsilon_t$, determined by the guessing advantage threshold.
Two approaches are possible.
The first one is to calculate a single $\epsilon_t$ for all events in the log. The second approach is to estimate an $\epsilon_t$ that minimizes the noise injected into each event individually. We use the second one.

We estimate $\epsilon_t$ for an event based on the group of prefix that ends with the event and the suffixes that starts with the event. Consequently, we adopt a \textit{personalized differential privacy} mechanism that uses a different $\epsilon_t$ value for every event. Personalized DP assumes that every individual has independent privacy specification $\Phi$~\cite{jorgensen2015conservative}. In other words, $\Phi_t =\{ (u_1,\epsilon_1), (u_2,\epsilon_2), ...., (u_m,\epsilon_m) \}$, where an individual $u\in U$, and the database has $m$ users.

\begin{defn}[Personalized Differential Privacy~\cite{jorgensen2015conservative}]\label{def:PDP}
Given a universe of cases $U$, a mechanism $M$ is said to be $\Phi$-personalized differentially private ($\Phi$-PDP) if all the event logs $L_1$ and $L_2$ differing at most on one value of a tuple $t$, and for all $S \subseteq Range (M)$:

\[ Pr[M(L_1) \in S] \leq exp(\Phi_t^u) \times Pr[M(L_2) \in S] \]

\noindent ...where $u \in U$ is the case that corresponds to tuple $t$, and $\Phi_t^u$ denote $u$'s privacy quantification.
\end{defn}




Note that differential privacy (cf.\ Defs~\ref{def:udp}, \ref{def:bdp}, and~\ref{def:bdp_attribute}) is a special case of personalized differential privacy (Def~\ref{def:PDP}), where all the $\epsilon$ values are the same~\cite{jorgensen2015conservative}.

In some settings, the PDP assumes that the $\epsilon$ values ($\Phi^u$) are given by the users. In this paper, we estimate the $\epsilon$ value for each event, i.e., we estimate $\Phi^u$.
We assume that the log publisher provides the maximum acceptable increase in the successful guessing probability after publishing the anonymized event log, called Guessing Advantage $\delta$. Laud et al.~\cite{laud2020framework} provide a framework that quantifies the $\epsilon$ value from a given guessing advantage $\delta$. They define the guessing advantage  
as the difference between the posterior probability (after publishing $M(L)$) and the prior probability (before publishing $M(L)$) (prior knowledge) of an attacker making a successful guess in $H_p$. Let $\delta$ be the maximum allowed guessing advantage, stated by the event log publisher.

\begin{defn}[Guessing Advantage]\label{def:ga}
Attacker's advantage in achieving the goal $h$ with precision $p$ is at most $\delta$ if, for any published event log $L'$, treating $L$ as a random variable,
%
\[Pr[h(L) \in H_p\ |\ M(L)=L'] - Pr[h(L) \in H_p] \leq \delta.\]
\end{defn}

Laud et al.~\cite{laud2020framework} quantify the maximum $\epsilon$ value that achieves the upper bound $\delta$ for every data instance.

\begin{prop}\label{prop:min_eps}
The maximum possible $\epsilon$ (i.e., the minimum noise) that achieves the upper bound $\delta$, w.r.t the above attack model, and fulfills the requirement~\ref{int:req:time} is
%
\begin{equation}\label{eqn:epsilon_time}
\footnotesize
\epsilon_{k}= - \ln\left( \frac{P_{k}}{1-P_{k}}(\frac{1}{\delta+ P_{k}} -1)\right) \cdot \frac{1}{r}\enspace,
\end{equation}
where $r$ is the maximum value in the range of values 
, and $P_k$ is the prior guessing probability for an instance $k$.
\end{prop}
\begin{proof}
The proof of Prop.~\ref{prop:min_eps} is available in Appendix.~\textbf{A.2}
\end{proof}

To quantify $\epsilon$ using Eq~(\ref{eqn:epsilon_time}), we need to study the distribution of values among every group of prefixes/suffixes. 
We quantify $\Phi_t$ values using the DAFSA-annotated event log to mitigate singling out an individual by their prefixes (or suffixes).
Given 
Definitions~\ref{def:dafsa}, and ~\ref{def:state_el}, all cases that share a common prefix will traverse a given state $s$ in the DAFSA corresponding to this prefix. The same holds for cases that share a common suffix. 

To anonymize the timestamps, we anonymize two components: (1) the start time of the case, which we define as the time difference between the start time of the case and the first start time in the event log); and (2) the execution time of each activity in the case, defined as the  difference between its execution timestamp and the timestamp of the successor activity in the same case. 

To anonymize the start time of each case, we group all the start times of the cases in the log to anonymize the fact that a given case happened on a specific day. Accordingly, we group events that have source state $s_0$ as a single group. To anonymize  the execution time of activities, we group the events that have the same source state, activity label, and target state in the DAFSA. For each group of the above, we use Eq~(\ref{eqn:p_k}) and~(\ref{eqn:epsilon_time}) to estimate a different $\epsilon_t$ value for every event ($\Phi_t$). 
Eq~(\ref{eqn:epsilon_time}) provides the maximum $\epsilon_t$ to fulfill~\ref{int:req:time}. The $\epsilon_t$ estimated for each event 
are shown in Table~\ref{tbl:filtered_log} (column $\epsilon_t$).


\subsubsection{Estimating $\epsilon_d$}

Parameter $\epsilon_d$ determines the noise to be applied to the occurrence count of each case variant in the log to prevent attack $h_1$. To estimate $\epsilon_d$ for a given group of prefixes or suffixes, we consider each group's size (count). This count is given by the following SQL query:

{
\begin{verbatim}
CT = SELECT SourceState, ActivityLabel, TargetState, COUNT(*)
FROM StateAnnotatedEventLog 
GROUP BY SourceState, ActivityLabel, TargetState
\end{verbatim}
}

The output of the above query is the \textit{DAFSA transition contingency table} (CT). We use the CT to estimate $\epsilon_d$ by using Eq~(\ref{eqn:p_freq}) and~(\ref{eqn:epsilon_time}). The $\epsilon_d$ of all the DAFSA transitions of the CT is the same, but the noise is drawn for every transition independently. The privacy proofs of $\epsilon$ estimation are presented in Appendix~\textbf{A.1}.

\begin{defn} [DAFSA Transitions Contingency Table]
A DAFSA transition contingency table $C$ is the histogram of counts for each transition $t=(s_i, a, s_e)$ of the DAFSA $D$, where $s_i$ is the source state, $a$ is the activity label, and $s_e$ is the target state of $t$.
\end{defn}

%

%
        
        
        

%
Table~\ref{tbl:dafsa_noise}~ shows the DAFSA transitions CT. These counts are called marginals. The marginals contain the correlations counts of the common sets of prefixes and suffixes of the DAFSA. 
We anonymize the marginals to prevent singling out that an individual has been through an activity, using the prefix and the suffix set of activities. 
Barak et al.~\cite{barak2007privacy} use an unbounded $\epsilon$-DP mechanism (cf.\ Definition~\ref{def:udp}) for marginal anonymization. Likewise, we use unbounded $\epsilon$-DP to anonymize the DAFSA transitions CT.

\begin{defn} [Differentially Private DAFSA Transitions Contingency Table]
\label{def:dp_dafsa}
Let $f$ be a query function that computes a DAFSA transitions contingency table with a set of transitions $t=(s_i, a, s_e)$ and a count cell $c_i$ for each transition. Let $M_f$ be an unbounded $\epsilon$-differentially private mechanism (by Definition~\ref{def:udp}) that injects noise into the result of $f$. A differentially private DAFSA transitions contingency table is defined as
$M(C):=\{(t_1,M_f(c_1)),(t_2,M_f(c_2)) ...,(t_n,M_f(c_n))\}$, where $n$ is the number of transitions.
\end{defn}

\subsection{ Case Sampling}

\label{sec:app:oversample}

Given the mechanism $M_f$ of Def.~\ref{def:dp_dafsa}, we need
to translate the noise estimated by the CT to the log.
Kifer et al.~\cite{kifer2011no} study the anonymization of CTs. They present the notion of a \textit{move} to define the anonymization of  the marginals of CTs. A move is a process that adds or deletes a tuple from the CT. 
We calculate the number of moves by drawing a random noise using $\epsilon_d$. 
Then, we translate the moves to a sampling of the cases that go through the same DAFSA transitions. 
Sampling may require zero or more moves.

The conference version of this article~\cite{elkoumy2021mine} has a stricter version of~\ref{int:req:trace}. It fulfills the requirement that ``the anonymized log must have the same set of case variants as the original log.''
We defined \textit{oversampling} as increasing a count (positive move) in the CT by replicating a random tuple in the log to fulfill such a requirement. 

\begin{defn} [DAFSA Transition Oversample]
\label{def:oversampling}
Given a DAFSA transition contingency table $C_i$, an oversample $O$ is a transformation that adds a  DAFSA transition instance to $C_i$, producing a contingency table $C_j=O(C_i)$, with an increase of only one count cell by 1. 

\end{defn}

This article considers a relaxation of the above requirement (cf.~\ref{int:req:trace}). We extend our approach to use sampling instead of oversampling, i.e., our approach performs both replication and deletion of cases from the log. We define \textit{sampling} as an increase or a decrease of a count in the CT.

\begin{table}[hbtp]

\scriptsize
\centering
\caption{DAFSA Transitions Contingency Table, and the generated Random Noise every DAFSA Transition for $\delta$=0.3 and estimated $\epsilon_d$=1.238}

    \begin{tabular}[t]{|c|c|c|c|c|}
        
        \hline
        
        \shortstack{Source\\State}     &     Activity     &     \shortstack{Target\\State}      & Count & Noise \\\hline
        
        $s_0$  &  A  &  $s_5$  & 3  & 0  \\\hline
        $s_5$  &  B  &  $s_2$  & 3  & 2 \\\hline
        $s_2$  &  C  &  $s_3$  & 5  & 1 \\\hline
        $s_0$  &  D  &  $s_4$  & 2  & -3 \\\hline
        $s_4$  &  A  &  $s_5$  & 2  & 0 \\\hline
        $s_5$  &  E  &  $s_2$  & 2  & 0 \\
         \hline
    \end{tabular}

\label{tbl:dafsa_noise}

\end{table}

The randomly generated sample size (based on $\epsilon_d$) can have a positive or a negative size. We draw a random noise value from the Laplace distribution $Lap(\Delta f/\epsilon_d)$, for each DAFSA transition. Table~\ref{tbl:dafsa_noise} shows the random noise for every DAFSA transition with the estimated $\epsilon_d$= 1.238. The positive size is translated to replicating a random tuple in the log. The negative size is translated to deleting a random tuple in the log. Adding or removing a prefix/suffix of a trace to the log affects a single frequency count by $1$, so $\Delta f = 1$.

\begin{defn} [DAFSA Transition Sampling]
\label{def:sampling}
Given a DAFSA transition contingency table $C_i$, a sample $Sample$ is a transformation that adds or deletes a  DAFSA transition instance from $C_i$, producing a contingency table $C_j=Sample(C_i)$, with an increase or a decrease of only one count cell by 1. 

\end{defn}

To 
avoid inserting new trace variants in the event log (cf.~\ref{int:req:trace}), we sample the prefix and the suffix of the sampled transition, i.e., we sample an entire case that goes through the sampled transition.

\begin{defn} [Case Sample]
\label{def:cases_oversampling}
Given an event log $L$, a case sample $Sample_{c_i}$ is a transformation that either duplicates or deletes a case $c_i$ of log $L$, that goes through a sampled DAFSA transition t, in such a way that 
it duplicates or deletes all the activities of $c_i$.

\end{defn}


\noindent\textbf{Timestamp Noise Injection}
At this step, we have the DAFSA annotated log, with the sampled case instances and an $\epsilon_t$ value of each event. First, if some cases have been replicated, we divide the $\epsilon_t$ value of the replicated cases by the number of replications, as it is considered repeating the same query more than once~\cite{dwork2014algorithmic}. Then, we draw a random noise from the Laplace distribution $Lap(\Delta f/\epsilon_t)$ to anonymize the relative time for every activity instance and the start time of each case. 
Finally, we transform 
the relative execution time of activities 
to timestamps.



Algorithm~\ref{alg:oversampling} presents the steps we perform in order to calculate the differentially private event log. This algorithm does not introduce new case variants so as to fulfill~\ref{int:req:trace}. 
Algorithm~\ref{alg:oversampling} starts by performing event log state-annotation as described in Sect.~\ref{sec:app:prep}. Then, we estimate the prior knowledge as described in Sect.~\ref{sec:app:prior}. Next, we perform case filtering as presented in Sect.~\ref{sec:app:filtering}. Then, we estimate $\Phi_t$, and $\epsilon_d$ as described in Sect.~\ref{sec:app:eps}. Then, we construct a correspondence (lookup) table between the 
DAFSA annotated  log and the case variants (line 7). This table maps every DAFSA transition to the case variants that traverse it. Also, we use this lookup table to track the updates over transitions. 
Second, we independently draw a random noise from the Laplace distribution $Lap(\Delta f/\epsilon_d)$ (line 8) for every transition. We initialize added noise counter to be zero (line 9). Next, we count the DAFSA transitions that need noise injection and their needed noise (lines 10). Next, we choose a random transition that needs noise (with their occurrence frequency as sampling weights) (line 12). Then, we randomly choose a case variant that goes through the chosen transition (with sampling weights of their number of instances) (line 13).
Next, based on the  noise (line 15), we either replicate (lines 16-17) or delete (lines 20-21) the chosen case variant by a number of times equals the  noise.
For every replication/deletion, we choose a random case variant instance from the  log to be replicated. 
We repeat this process until all the transitions have the minimum required noise. Finally, we divide $\Phi_t$ of the replicated cases by the number of replications (line 26).
By the end, the event log is anonymized by $\Phi=\{\epsilon_d, \Phi_t \}$ values calculated from the input maximum guessing advantage $\delta$.
\SetNlSty{}{}{.}

\begin{algorithm}
\scriptsize 
\hspace*{\algorithmicindent} \textbf{Input:}~$L$:~Event~Log, $\delta$: Guessing~Advantage~Threshold\\
 \hspace*{\algorithmicindent} \textbf{Output:}~$L'$: $\Phi$-Personalized~Differentially~Private~Event~Log \\
 $L$= eventLogStateAnnotation($L$)\;
 $p_k$=priorKnowledgeEstimation($L$,$\delta$)\;
 $L_f$= caseFiltering($L$, $p_k$)\;
 $\epsilon_d$, $\Phi_t$= $\epsilon$Estimation($L_f$,$p_k$, $\delta$)\;
 DafsaLookup $=$ Build DAFSA annotated event log to case variant lookup($L_f$) \;
 
 DafsaLookup[i].neededNoise$=z_i$, where $z_i$ is sampled from $Lap(\Delta f/\epsilon_d)$ independently for every transition $t_i$\;
 DafsaLookup .addedNoise=0\;
 cnt = count($|$DafsaLookup .addedNoise$|$ $<$ $|$DafsaLookup.neededNoise$|$) \;
 $L'$ = $L_f$\;
 
 \While{ cnt $>$ 0} {
  selectedTransition = pick a random transition such that $|$DafsaLookup.addedNoise$|$ $<$ $|$DafsaLookup.neededNoise$|$\;
  
  pickedTraces = pick $x$ random traces that traverse selectedTransition, where $x$ = selectedTransition.neededNoise\;
  
  \ForEach {$t \in$  pickedTraces  } {
  
  \uIf{DafsaLookup[t].neededNoise $>$ 0}
    {  DafsaLookup[t].addedNoise $++$ \;
     add to $L'$ a replica of a random case with a case variant $=$ t \;
    }
  \Else{
      DafsaLookup[t].addedNoise $--$ \;
      delete from $L'$ a random case with a case variant $=$ t \;
    }

  }
  cnt = count($|$DafsaLookup .addedNoise$|$ $<$ $|$DafsaLookup.neededNoise$|$) \;
  
 }
 
 DafsaLookup[i].$\Phi_{t_i}$=DafsaLookup[i].$\Phi_{t_i}$/DafsaLookup[i].NumOfReplicas, where $i$ is the replicated cases\;
 $L'$.timestamp= $L'$.timestamp + $z_i$, where $z_i$ is sampled from $Lap(\Delta f/$DafsaLookup[i].$\epsilon_{t_i}$)
 
\Return{ $L'$}
 
 \caption{Event Log Anonymization Algorithm}
 \label{alg:oversampling}

\end{algorithm}


Given $\Phi=\{\epsilon_d, \Phi_t \}$, estimated from the  parameter $\delta$, below we show that the output of Algorithm~\ref{alg:oversampling} ensures the $\Phi$-personalized differential privacy guarantees. 

\begin{lem}
\label{lem:dpel}
Let $L$ and $\delta$ be an event log and a maximum guessing advantage threshold. Log $L'$ = Alg1($L$,$\delta$) fulfills the two properties in Def~\ref{def:diff_priv_EL} (i.e.,\ 
$L'$ is $\Phi$-personalized differentially private), and $L'$ fulfills requirements~\ref{int:req:trace} and~\ref{int:req:time}.
\end{lem}
\begin{proof}
See Appendix~\textbf{A} and Lemma~\textbf{A.8}.
\end{proof}

Algorithm~\ref{alg:oversampling} anonymizes the control-flow  and the timestamps using two separate mechanisms. The timestamp anonymization mechanism does not alter the order of activity instances within a trace (i.e. it does not alter the control-flow). In contrast, the control-flow anonymization mechanism sometimes affects the timestamps. For example, when a trace is duplicated for over-sampling purposes, this affects the $\epsilon$ guarantee, because we are effectively releasing twice the same anonymity trace. Hence, if a trace is duplicated $D$ times, the resulting release would be $(D \cdot \epsilon)$-differentially private~\cite{dwork2014algorithmic}.
Therefore, to ensure that the release is in fact $\epsilon$-differentially private, we first divide the $\epsilon$ for timestamp anonymization by the maximum number of duplications $\mathit{MaxD}$, and then we inject the noise required to achieve $(\epsilon/\mathit{MaxD})$ differential privacy (line 26). 
Given the above, the combination of the control-flow  anonymization mechanism and the timestamp anonymization mechanism, as presented in Algorithm~\ref{alg:oversampling}, does not invalidate either of the two differential privacy guarantees.

\subsection{Timestamp Compression}

\label{sec:app:post}
%

This article assumes that the start timestamps of the first and the last case in the original log are public. 
However, the timestamp anonymization introduces time shifts, 
making 
some cases happen before the original first timestamp or after the last timestamp of the log.
To reduce the impact of time shifts over the log, we perform timestamp compression as post-processing of the anonymized log. 



\begin{prop} [Differential Privacy under Post-processing~\cite{dwork2014algorithmic}]
\label{prop:postprocessing}
A post-processing algorithm $P$ of an event log $L$  gives $\Phi$-personalized differential private event log, if and only if it has been applied to the output of an algorithm $A$ that gives an $\Phi$-personalized differential private event log.

\end{prop}
\begin{proof}
The proof of Prop.~\ref{prop:postprocessing} is in~\cite{dwork2014algorithmic} (cf. Proposition 2.1)
\end{proof}

Given the $\Phi$-PDP log, 
 the post-processing of $M(L)$ output is differentially private. 
We use the public information, the start timestamp of the first and the last cases in the event log, to post-process the anonymized log. We perform timestamp compression to make the cases' timestamp fall between the original start timestamp of the first and the last cases in the  log. We multiply the relative time values with a compression factor. We define the compression factor as:

\begin{equation}\label{eqn:cmprs_factor}
\scriptsize
Compression~Factor = \frac{Original~Range}{Anonymized~Range + Orignal~Range} * \frac{1}{2},
\end{equation}
where the original range is the difference in days between the start timestamp of the first and the last cases in the original event log, and the anonymized range is the difference in days between the start timestamp of the first and the last cases in the anonymized event log. Table~\ref{tbl:anonymized_log} shows the anonymized version of the input event log in Table~\ref{tbl:event_log} with a guessing advantage threshold $\delta$=0.3.

Lastly, we generate new Case IDs for the anonymized log cases that do not link to the original log. Also, we reorder the events in the log by their new timestamps.

\begin{table}[hbtp]

	\centering
\caption{Differentially Private Event Log with $\delta$=0.3}

\scriptsize
	\begin{tabular}[t]{|c|c|c|}
	
\hline

Case ID	&	Activity	&	Timestamp	\\ \hline
66d19fc868978d2fc1e	&	D	&	2020-08-08 04:26:24	\\
66d19fc868978d2fc1e	&	A	&	2020-08-08 07:03:57	\\
66d19fc868978d2fc1e	&	E	&	2020-08-08 07:43:20	\\
66c81c1d1a9773464aa	&	D	&	2020-08-09 12:16:42	\\
66d19fc868978d2fc1e	&	C	&	2020-08-09 13:33:07	\\
66c81c1d1a9773464aa	&	A	&	2020-08-09 14:54:15	\\
66c81c1d1a9773464aa	&	B	&	2020-08-09 15:53:27	\\
c4247e4c1e9292166cd	&	A	&	2020-08-09 18:44:10	\\
c4247e4c1e9292166cd	&	E	&	2020-08-09 19:23:33	\\
efd62407e7f2b016e33	&	A	&	2020-08-09 19:34:35	\\
efd62407e7f2b016e33	&	B	&	2020-08-09 20:00:23	\\
efd62407e7f2b016e33	&	C	&	2020-08-10 00:49:12	\\
66c81c1d1a9773464aa	&	C	&	2020-08-10 04:41:39	\\
64b00e4dfc2b2145e17	&	D	&	2020-08-10 05:17:48	\\
64b00e4dfc2b2145e17	&	A	&	2020-08-10 07:55:20	\\
64b00e4dfc2b2145e17	&	B	&	2020-08-10 08:44:53	\\
64b00e4dfc2b2145e17	&	C	&	2020-08-11 00:59:32	\\
c4247e4c1e9292166cd	&	C	&	2020-08-11 06:16:25	\\
cde714d92c42217500a	&	A	&	2020-08-11 07:02:37	\\
cde714d92c42217500a	&	B	&	2020-08-11 07:33:55	\\
cde714d92c42217500a	&	C	&	2020-08-11 15:26:28	\\
\hline

	\end{tabular}
\label{tbl:anonymized_log}

\end{table}


\section{Evaluation}
\label{sec:eval}
%


To address the problem stated in Sect.~\ref{sec:intro} under requirements~\ref{int:req:trace} and~\ref{int:req:time}, the proposed method injects differentially private noise in two ways: (i) by sampling and filtering some of the traces in the log; and (ii) by altering the event timestamps. The noise injection and case filtering affect the utility of the anonymized logs. We measure the effect of anonymization on the utility by comparing the anonymized logs against the original ones. We compare the performance of different design choices.
Also, we compare the proposed approach against the state-of-the-art.

Accordingly, we define the following evaluation questions:
\begin{enumerate}[label=\textbf{EQ\arabic*.}]
    \item \label{res:rq:filtering} What is the effect of case filtering and sampling on the output utility?
    \item \label{res:rq:baseline} Does the proposed approach outperform the state-of-the-art baselines in terms of the output utility?
    \item \label{res:rq:performance} What is the difference between different design choices and the state-of-the-art in terms of computational efficiency?
\end{enumerate}

\subsection{Evaluation Measures}


Given a log, a typical output of process mining tools is the DFG. To measure the utility loss of the anonymization on the DFG, we compare the DFG resulting from the anonymized log against the original one. To measure the difference between two DFGs, we use the \textit{Earth Movers' Distance} (EMD)~\cite{ramdas2017wasserstein}.
The EMD has been used as a log comparison metric metric in several process mining studies such as in~\cite{rafiei2021group,rafiei2020towards,LeemansABP21}.
The EMD between two distributions $u$ and $v$ is the minimum cost of transforming $u$ into $v$. The cost is the distribution weight that needs to be moved, multiplied by the distance it needs to move. Formally:
\begin{equation}
\footnotesize
EMD (u,v)= \inf_{\pi \in \Gamma(u,v)} \int_{\mathbb{R}\times \mathbb{R}} |x - y | d\pi(x,y),
\end{equation}
 where $\Gamma(u,v)$ is the set of distributions on $\mathbb{R}\times \mathbb{R}$ whose marginals are $u$ and $v$.

Increasing the log size makes it unpractical to perform anonymization within sufficient execution time~\cite{elkoumy2021privacy}. Hence, the privacy computational models should be scalable in order to preserve privacy practically.
Consequently, we conduct a wall-to-wall run time experiment to assess the efficiency of the method. We measure the time between reading the input XES file and the generation of its anonymized version.




\subsection{Datasets}

To answer our evaluation questions, we rely on the real-life event logs publicly available at 4TU Centre for Research Data\footnote{\url{https://data.4tu.nl/}} as of February 2021. 
We considered the logs mentioned in Appendix~\textbf{B} Table~\textbf{1}.
The selected logs contain the process execution of different domains, e.g., government and healthcare.
From the set of available logs, we excluded the event logs that are not business processes (e.g., ``Apache Commons'', ``BPIC 16'' logs, ``Junit 4.12"). 
Also, we exclude the set of event logs ``coSeLog" as they are a pre-processed version of BPI challenge 15. 
Finally, we select a single log for each set of logs in BPIC 13, 14, 15, 17, 20.

\subsection{Experiment Setup} 

\label{sec:exp_setup}
We implement the proposed model as part of a prototype, namely Amun\footnote{\url{https://github.com/Elkoumy/amun}}. 
We run the experiment on a single machine with  AMD Opteron(TM) Processor 6276 and 32 GB memory. We time out any experiment at 24 hours.
Also, in our experiment, we consider only the end timestamp to calculate the relative time of an event for simplicity, and the same approach is still valid to apply DP. Further, we keep only the three attributes in every event log: case ID, Activity, and timestamp.

We evaluate and compare the different design choices of our approach. We evaluate the oversampling presented in the conference version of this article~\cite{elkoumy2021mine}, the proposed approach using all the proposed steps in Sect.~\ref{sec:approach} (filtering risky cases then sampling), and the proposed approach without the third step (sampling without filtering). We use the  EMD to compare the anonymized event log against the original for all the design choices.

We compare the proposed approach  against the state-of-the-art. The studies that consider PPPM from the case perspective are~\cite{mannhardt2019privacy,fahrenkrog2020pripel,fahrenkrog2021sacofa,rafiei2020tlkc,rafiei2021group}. In our comparison, we do not include the work in~\cite{rafiei2020tlkc,rafiei2021group} because the parameters' interpretation of the k-anonymity privacy model is different from the DP model.
The studies~\cite{mannhardt2019privacy,fahrenkrog2020pripel,fahrenkrog2021sacofa} adopt DP.
Mannhardt et al.~\cite{mannhardt2019privacy} anonymize two types of queries: the query ``frequencies of directly-follows relations'' and ``frequencies of trace variants''. The output of the anonymization of~\cite{mannhardt2019privacy} is not an event log.
PRIPEL~\cite{fahrenkrog2020pripel} anonymizes the event log while adopting the trace variant queries anonymization that has been proposed in~\cite{mannhardt2019privacy}. We compare the proposed approach against~\cite{fahrenkrog2020pripel}. SaCoFa~\cite{fahrenkrog2021sacofa} anonymizes the case variant queries. The output is an event log without the time attribute. We include SaCoFa~\cite{fahrenkrog2021sacofa} only in the EMD frequency experiments because it does not consider timestamp anonymization. PRIPEL and SaCoFa take three input parameters, namely $\epsilon$, k, and N. 
To select the parameters' values, we run several experiments for different values of the pruning parameter k (0.5\%, 1\%, and 5\% of the cases), and we select the best results. For each parameter value, we run 10 experiments, and we take the average value. For the maximum trace length, we set $N$ to the average trace length of the log. Other experiments with different parameter values can be found in Appendix~\textbf{B}.

PRIPEL accepts a single $\epsilon$ value for both the trace variant anonymization and the timestamp attribute anonymization, and SaCoFa adapts a single $\epsilon$ value for all the events. On the contrary, our proposed approach uses different $\epsilon$ values ($\Phi$). Consequently, we use the average, the minimum, and the maximum  values of the estimated $\Phi$ values in our approach as input to PRIPEL and SaCoFa and we evaluate the output using the EMD. We use SaCoFa with only the  frequency  EMD because it anonymizes the case variant query only.  Below, Table~\ref{tbl:emd} presents the average estimated $\epsilon$ values. Tables B.2 and B.3 in Appendix~\textbf{B} present the minimum and maximum $\epsilon$ values. All the anonymized logs are available in Appendix~\textbf{B}.

\subsection{Results}
\label{sec:results}

\begin{table*}[hbtp]

	\caption{Earth Movers' Distance for the output of different anonymization approaches using the average value of $\Phi$. A ``-'' means that the approach ran out of memory or timed out.}
	\label{tbl:emd}
\scriptsize

\centering
	\begin{tabular}{ p{0.8cm} p{0.2cm} p{0.6cm} c c c c cc c c c   }
		\toprule
        \multirow{2}{*}{Log}		&	\multirow{2}{*}{$\delta$} 	&	\multirow{2}{*}{avg($\Phi$)}	&	 \multicolumn{4}{c}{EMD Freq} &	 \multicolumn{4}{c}{EMD Time}	\\
       
        	\cmidrule(lr){4-8} \cmidrule(l){9-12}	&		&		&	   $Amun_s$	&	  $Amun_f$	&	  $Amun_o$	&	  $PRIPEL$ & $SaCoFa$	&	  $Amun_s$	&	  $Amun_f$	&	 $Amun_o$	&	 $PRIPEL$	\\
        \hlineB{2}
        
        \multirow{3}{*}{BPIC12}	&	0.2	&	1.48	&	\textbf{331.02}	&	653.36	&	2301.62	&	946.9 &1006.60	&	40.30	&	\textbf{8.08}	&	192.13	&	25.75	\\
        	&	0.3	&	2.00	&	\textbf{212.64}	&	742.39	&	1597.79	&	966.99	&969.29& 	20.32	&	\textbf{13.54}	&	107.96	&	26.7	\\
        	&	0.4	&	2.47	&	\textbf{142.37}	&	785.87	&	1275.43	&	966.88	&876.56&	\textbf{12.68}	&	16.74	&	72.56	&	26.7	\\\hline

        \multirow{3}{*}{BPIC13}	&	0.2	&	1.49	&	1131.91	&	\textbf{1053.45}	&	7613.27	&	3771.09	&2992.09&	778.18	&	811.51	&	3343.92	&	\textbf{197.41}	\\
        	&	0.3	&	1.99	&	840.55	&	\textbf{258.45}	&	5417.64	&	3792.55&2562.55	&592.36	&	307.26	&	2055.23	&	\textbf{197.38}	\\
        	&	0.4	&	2.44	&	\textbf{558.45}	&	2450.45	&	4296.73	&	3775.82	&2623.18&	486.86	&	\textbf{75.13}	&	1751.96	&	197.19	\\\hline

        \multirow{3}{*}{BPIC14}	&	0.2	&	1.25	&	429.60	&	\textbf{395.24}	&	1905.69	&	531.42	&494.34&	132.53	&	130.00	&	577.32	&	\textbf{10}	\\
        	&	0.3	&	1.74	&	298.86	&	\textbf{281.43}	&	1333.95	&	-	&488.67&	82.02	&	\textbf{76.37}	&	321.72	&	-	\\
        	&	0.4	&	2.27	&	208.55	&	\textbf{179.94}	&	1012.62	&	-	&385.8&	\textbf{47.12}	&	50.61	&	178.08	&	-	\\\hline

        \multirow{3}{*}{BPIC15}	&	0.2	&	0.42	&	20.71	&	\textbf{18.74}	&	80.61	&	-	&20.84&	5.68	&	\textbf{4.18}	&	22.47	&	-	\\
        	&	0.3	&	0.69	&	14.82	&	\textbf{7.93}	&	52.42	&	10.71	&10.84&	3.15	&	2.28	&	11.17	&	\textbf{0.8}	\\
        	&	0.4	&	0.99	&	12.60	&	\textbf{2.79}	&	39.09	&	-	&10.30&	2.46	&	\textbf{1.05}	&	8.87	&	-	\\\hline

        \multirow{3}{*}{BPIC17}	&	0.2	&	1.81	&	\textbf{141.37}	&	1925.92	&	1159.08	&	2454.47	&1601.54&	117.56	&	\textbf{75.98}	&	268.08	&	109.10	\\
        	&	0.3	&	2.21	&	\textbf{78.60}	&	2667.91	&	938.79	&	-	&1271.89&	\textbf{95.07}	&	131.76	&	206.76	&	-	\\
        	&	0.4	&	2.66	&	\textbf{49.93}	&	2674.76	&	938.79	&	-	&1215.84&	\textbf{79.33}	&	133.77	&		206.76&	-	\\\hline

        \multirow{3}{*}{BPIC18}	&	0.2	&	0.91	&	3775.468&	\textbf{659.01}&	17668.46
					&	-	&3555.97&	2176.17 &	\textbf{179.39}	& 3206.75
					&	-	\\
        	&	0.3	&	1.31	&	2922.80	&	\textbf{1752.72}	&	12063.81	&	-&3763.41	&	1319.69	&	\textbf{325.14}	&	2201.39	&	-	\\
        	&	0.4	&	1.72	&	\textbf{2372.08}	&	2812.65	& 9055.04		&	-&3789.32	&	885.81	&	\textbf{517.59}	&	1647.15	&	-	\\\hline

        \multirow{3}{*}{BPIC19}	&	0.2	&	2.96	&	946.99	&	\textbf{811.80}	&	4509.08	&	-	&1037.15&	\textbf{524.18}	&	608.02	&	1513.92	&	-	\\
        	&	0.3	&	3.51	&	743.96	&	\textbf{572.71}	&	3094.05	&	-	&936.21&	\textbf{491.84}	&	498.59	&	1054.23	&	-	\\
        	&	0.4	&	4.00	&	612.48	&	\textbf{399.36}	&	2376.21	&	-	&872.76&	500.57	&	\textbf{360.54}	&	751.67	&	-	\\\hline

        \multirow{3}{*}{BPIC20}	&	0.2	&	2.73	&	18.97	&	\textbf{18.69}	&	138.05	&	98.10	&142.92&	65.23	&	69.04	&	180.50	&	\textbf{23.94}	\\
        	&	0.3	&	3.27	&	14.88	&	\textbf{10.42}	&	99.76	&	-	&126.27&	44.46	&	\textbf{40.01}	&	109.90	&-	\\
        	&	0.4	&	3.75	&	10.55	&	\textbf{2.24}	&	76.61	&	-	&88.85&	41.70	&	\textbf{35.78}	&	82.67	&	-	\\ \hline

        \multirow{3}{*}{CCC19}	&	0.2	&	0.23	&	12.78	&	\textbf{3.70}	&	30.81	&	-	&-&	\textbf{0.00}	&	\textbf{0.00}	&	\textbf{0.00}	&	-	\\
        	&	0.3	&	0.54	&	4.55	&	\textbf{2.23}	&	25.77	&	-	&-&	\textbf{0.00}	&	\textbf{0.00}	&	\textbf{0.00}	&	-	\\
        	&	0.4	&	0.85	&	5.54	&	\textbf{3.21}	&	16.64	&	-	&-&	\textbf{0.00}	&	\textbf{0.00}	&	\textbf{0.00}	&	-	\\\hline

        \multirow{3}{*}{CredReq}	&	0.2	&	1.39	&	\textbf{0.00}	&	2.00	&	4.00	&	\textbf{0.00}	&0&	233.08	&	238.79	&	186.94	&	\textbf{0.00}	\\
        	&	0.3	&	1.80	&	\textbf{0.00}	&	\textbf{0.00}	&	3.00	&	\textbf{0.00}	&0&	203.94	&	182.39	&	220.17	&	\textbf{0.00}	\\
        	&	0.4	&	2.03	&	\textbf{0.00}	&	2.00	&	2.00	&	\textbf{0.00}	&0&	201.81	&	231.37	&	180.77	&	\textbf{0.00}	\\ \hline

        \multirow{3}{*}{Hospital}	&	0.2	&	0.21	&	81.40	&	\textbf{75.74}	&	406.45	&	-	&85.24&	\textbf{10.30}	&	11.33	&	55.44	&	-	\\
        	&	0.3	&	0.41	&	\textbf{57.91}	&	61.96	&	264.09	&	-	&84.86&	7.86	&	\textbf{7.42}	&	31.71	&	-	\\
        	&	0.4	&	0.60	&	\textbf{58.68}	&	63.93	&	210.40	&	-	&84.24&	6.07	&	\textbf{6.01}	&	20.70	&	-	\\ \hline

        \multirow{3}{*}{Sepsis}	&	0.2	&	1.31	&	56.84	&	\textbf{32.20}	&	427.64	&	118.02	&121.3&	8.97	&	\textbf{4.37}	&	61.50	&	8.68	\\
        	&	0.3	&	1.83	&	\textbf{28.46}	&	67.97	&	286.23	&	118.63	&114.23&	6.35	&	\textbf{3.60}	&	32.93	&	8.64	\\
        	&	0.4	&	2.37	&	\textbf{43.38}	&	101.64	&	232.50	&	118.71	&96.84&	\textbf{4.26}	&	6.48	&	26.61	&	8.64	\\ \hline

        \multirow{3}{*}{Traffic}	&	0.2	&	4.50	&	1.64	&	\textbf{0.90}	&	33.50	&	-	&36.67&	8250.42	&	\textbf{7253.29}	&	8627.86	&	-	\\
        	&	0.3	&	5.26	&	\textbf{0.61}	&	8.01	&	25.51	&	-	&32.90&	7767.32	&	\textbf{6720.80}	&	7081.06	&	-	\\
        	&	0.4	&	5.28	&	\textbf{0.00}	&	2951.50	&	21.00	&	-	&22.50&	\textbf{7182.96}	&	12788.57	&	7419.02	&	-	\\ \hline

        \multirow{3}{*}{Unrine.}	&	0.2	&	2.87	&	6.53	&	\textbf{5.53}	&	66.00	&	290	&25.60&	44.96	&	\textbf{39.34}	&	184.68	&	94.74	\\
        	&	0.3	&	3.41	&	\textbf{0.00}	&	2.87	&	53.67	&	290.47	&24.70&	\textbf{28.12}	&	33.27	&	103.67	&	94.73	\\
        	&	0.4	&	3.84	&	\textbf{1.47}	&	2.53	&	47.00	&	290.47	&2.60&	\textbf{21.89}	&	26.35	&	90.93	&	94.73	\\ 
\bottomrule

	\end{tabular}
	
\end{table*}


Table~\ref{tbl:emd} shows the results for both the frequency and time annotated DFG.
 A ``-'' indicates that the approach runs out of memory (32 GB) or times out (24 hours).  $avg(\Phi)$  refers to the average $\epsilon$ value (in $\Phi$) estimated by the proposed approach for anonymization. We use the directly-follows frequency between activities to annotate the frequency-annotated DFG.
We use the total relative time between two activities to annotate the time-annotated DFG. The time EMD distance is measured in terms of months. The best result for every input $\delta$ is 
in bold. 
$Amun_f$ outperforms other settings in both the frequency and time EMD in most of the logs because case filtering decreases the needed noise to anonymize the timestamp, and hence $Amun_f$ has a lower utility loss. 

The utility loss differs across logs. We see a decrease in the utility loss for structured logs such as Credit Requirement, Unrineweginfectie (urinary tract infection), and Traffic Fines. For example, the output of $Amun_s$ for Unrineweginfectie has a maximum frequency EMD of $6.53$, and for the same log, $Amun_f$ has a maximum EMD of $5.53$. This difference is because $Amun_f$ filters out cases to reduce the timestamp noise injection. PRIPEL has a frequency EMD of $290$, and SaCoFa has a frequency EMD of $25.6$. That happens due to infrequent case variants trimming.

With anonymizing unstructured event logs, which is more challenging due to the uniqueness of cases, we see an increase in utility loss. For instance, the output of $Amun_s$ for the Sepsis cases log has a maximum frequency EMD of $56.84$, and for the same log, $Amun_f$ has a maximum EMD of $101.04$. The increase of the utility loss is significant with $Amun_f$ because it filtered out $580$ case variants, in contrast to $Amun_s$, which dropped out only $87$ case variants. For the same log, PRIPEL has a maximum EMD of $118.71$. That happens because PRIPEL filtered out $839$ case variants. SaCoFa has a maximum EMD of $121.3$ because it filtered out 787 case variants and added new 19 case variants (false positives).


The oversampling setting ($Amun_o$) 
has the largest frequency and time EMD due to duplicating cases and dividing $\Phi_t$ by the number of duplications. PRIPEL has a time EMD that is close to the filtering settings. However, the used $\epsilon$ with PRIPEL is the average estimated $\Phi$. Thus, the proposed approach can provide similar or better time EMD with stronger privacy metrics. The effect of using the same $\epsilon$ in PRIPEL for both the case variant and timestamp anonymization appears in the frequency EMD. The proposed approach has a lower frequency EMD than PRIPEL and SaCoFa in all the logs. The same conclusion could be driven from Tables B.2 and B.3 in Appendix~\textbf{B} which contain the EMD of the approaches for the minimum and maximum $\epsilon$ values.

\begin{figure}
    \centering
    
       \subfigure[Unrineweginfectie Active cases over time]
    {
       \centering
         \includegraphics[width=1\columnwidth]{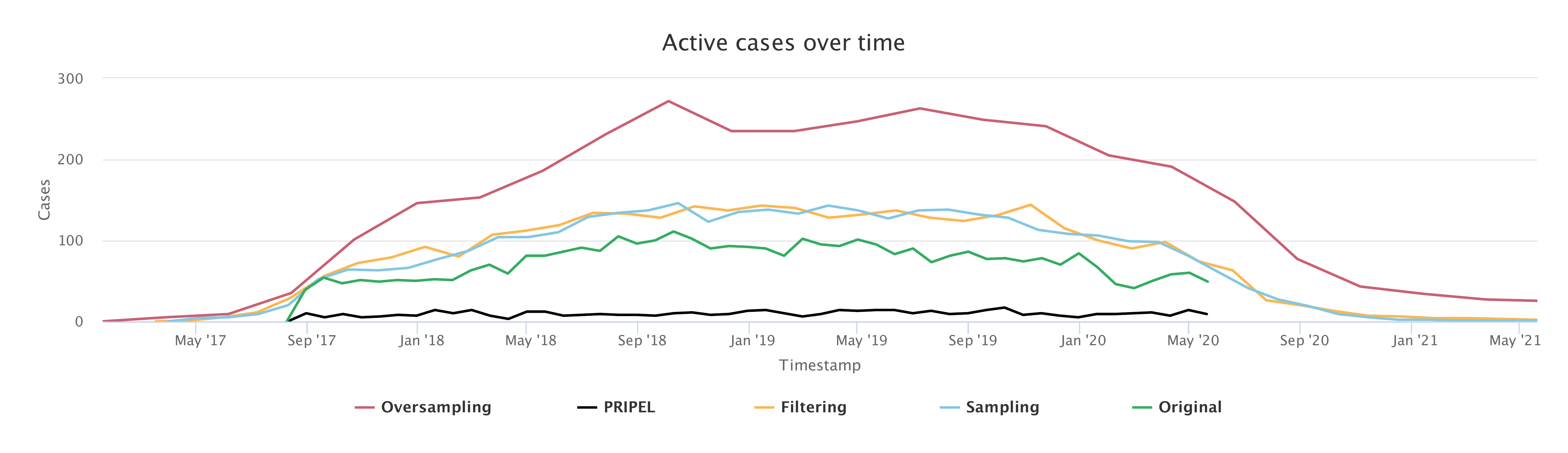}
          
         \label{fig:unrine_cases_over_time}
        
    }
    \\
    \subfigure[Sepsis Active Cases over Time]
    {
       \centering
       
         \includegraphics[width=1\columnwidth]{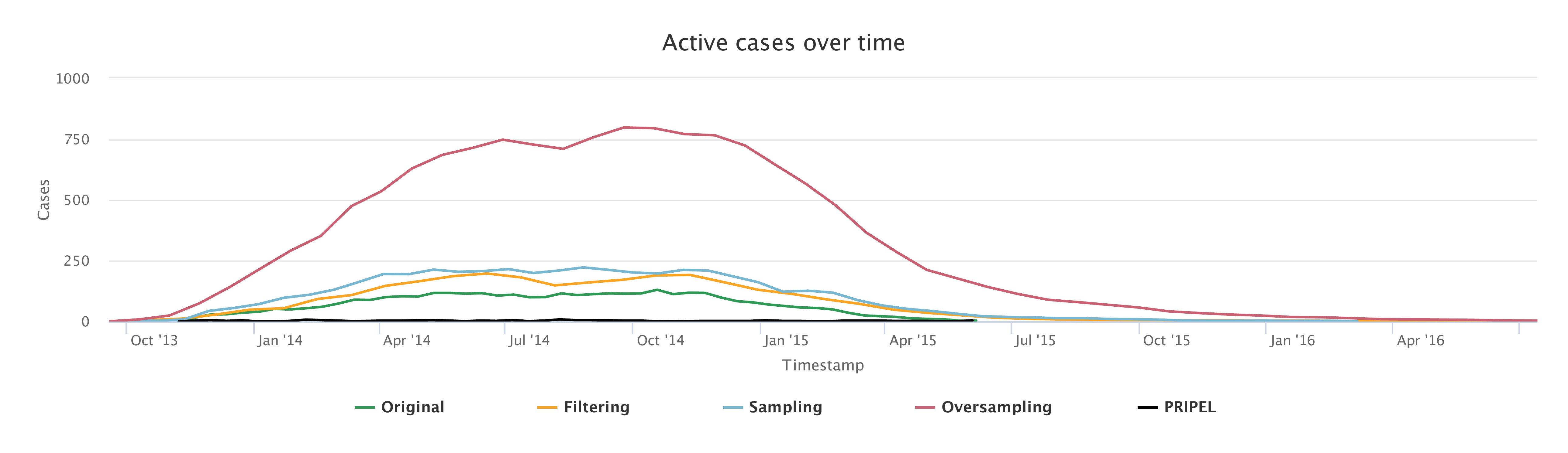}
         \label{fig:sepsis_cases_over_time}
         
    }
    
    \caption{Variant Analysis comparison between Unrine. and Sepsis  event logs and their anonymized versions, with $\delta=0.2$, average $\epsilon=1.31$ for Sepsis, and average $\epsilon=2.87$ for Unrine. The figures are zoomed by 70\%.}
        \label{fig:variant_analysis}
    
\end{figure}

We compare the variant analysis of different design choices for structured and unstructured  logs. 
Fig.~\ref{fig:variant_analysis} shows a variant analysis between Unrineweginfectie and Sepsis logs and their anonymized versions by different design choices and PRIPEL.
All the design choices have their best utility for the Unrineweginfectie log, which has 1650 cases and only 50 case variants.
Fig.~\ref{fig:unrine_cases_over_time} shows that $Amun_s$, $Amun_f$, and PRIPEL result in logs with close active cases over time to the original log. The false negatives in the anonymized logs are 1, 2 and 34 for $Amun_s$, $Amun_f$ and PRIPEL, respectively. However, $Amun_o$ adds more noise than other approaches, though it keeps the false negatives to zero due to oversampling.

Fig.~\ref{fig:sepsis_cases_over_time} shows the active cases over time for the sepsis cases event log, with 1050 cases and 846 case variants. For both $Amun_s$ and $Amun_f$ the anonymized Sepsis cases has a closer behavior to the original log. $Amun_o$ generates more noise with unstructured event logs than structured logs. The anonymized log by PRIPEL has fewer case variants than the original log. The false negatives in the anonymized logs are 37, 317, and 839 for $Amun_s$, $Amun_f$, and PRIPEL, respectively.

     

        

\begin{table}[hbtp]

	\caption{Execution time experiment. The time is measured in minutes for an input $\delta=0.2$. A ``-'' means that the approach ran out of memory or timed out.}
	
	\label{tbl:time}
\scriptsize
\centering
	\begin{tabular}{c  c c c c}
		\toprule
dataset	&		$Amun_f$ &	$Amun_s$	&	$Amun_o$&	PRIPEL	\\
\hline
BPIC12	&		4.44	&	7.38	&	12.50	&	24.35	\\
BPIC13	&		1.11	&	1.09	&	3.90	&	4.20	\\
BPIC14	&		44.57	&	43.83	&	101.40	&	-	\\
BPIC15	&		4.38	&	4.45	&	8.40	&	-	\\
BPIC17	&		6.88	&	10.80	&	17.40	&	1.46	\\
BPIC18	&		128.32	&	321.59	&	542.50	&	226	\\
BPIC19	&		96.89	&	52.73	&	87.50	&	-	\\
BPIC20	&		2.64	&	2.52	&	4.10	&	3.18	\\
CCC19	&		0.04	&	0.09	&	0.10	&	-	\\
CreditReq	&		1.21	&	1.35	&	1.20	&	14.50	\\
Hospital	&		11.27	&	10.65	&	34.10	&	-	\\
Sepsis	&		0.87	&	1.00	&	1.90	&	0.05	\\
Traffic	&		9.71	&	11.35	&	8.20	&	-	\\
Unrine.	&		0.14	&	0.16	&	0.20	&	0.02	\\

\bottomrule

	\end{tabular}
	
\end{table}

We conduct a wall-to-wall run time experiment to assess the efficiency of the method. We measure the time between reading the input XES file and the generation of its anonymized version. The results are reported in Table~\ref{tbl:time}, and the values are in minutes. The run time increases with case variants (as it contains more DAFSA transition groups). 
The execution times for logs with numerous events and with low $\delta$ values are in the order of hours, e.g., 2.13 hours for BPIC18 (2.5 million events) with $\delta$=0.2 because the noise injection algorithm iterates multiple times over each transition (lines 15-23 in Algorithm~\ref{alg:oversampling}), and the number of DAFSA states for this log is high (638,242 states). This shortcoming can be tackled via parallelization, as the privacy quantification over each DAFSA transition is independent of others. PRIPEL does not scale and runs out of memory for event logs that have numerous cases and events.
The above experiments were all done using one single thread to avoid bringing additional variables (number of computing nodes and cores) into the experiments.

The results in Table~\ref{tbl:emd} and Fig.~\ref{fig:variant_analysis} show that using different $\epsilon$ values for case variant and timestamp anonymization leads to  lower utility loss. Moreover, the use of personalized DP to estimate different $\epsilon$ per event leads to a stronger privacy guarantee with lower utility loss.

We acknowledge that the above observations are based on a limited population of logs (14). However, these logs were selected from a broader population of close to 50 real-life logs.


\section{Conclusion and Future Work}

\label{sec:conclusion}

This article proposed a concept of the differentially private event log and a mechanism to compute such logs. A differentially private event log limits the increase in the probability that an attacker may learn a suffix of an individual's trace given a prefix (or vice-versa) or the timestamp of activity in an individual's trace. To this end, we inject differentially private noise by sampling the traces in the log. This approach does not add case variants (cf. Sect.~\ref{sec:app:oversample}) and hence fulfills~\ref{int:req:trace}. To fulfill~\ref{int:req:time}, we quantify $\epsilon$ based on a technique that finds the maximum $\epsilon$ (minimum noise) that keeps the guessing advantage below $\delta$ (cf. Sect.\ref{sec:attack_model} and Prop.~\ref{prop:min_eps}). 

The evaluation show that the proposed approach is a step toward anonymizing event logs while preserving the utility of the process mining analysis. The proposed approach outperforms the baselines in terms of  earth movers' distance and generalization to all fourteen real-life  logs selected in the evaluation. Furthermore, the approach can process large size  logs in practical memory size (32 GB).

A limitation of the proposed method is that it anonymizes the timestamp without performing calendar anonymization, violating organizational rules (e.g., single activity at each timestamp). A possible avenue for future work is to model and include organizational rules while anonymizing logs. A second limitation is that the input log is assumed to have only three columns: case ID, activity, and timestamp. Real-world event logs 
contain other columns, e.g., resources. 
To do so, we need to extend the log representation, e.g., via multidimensional data structures instead of DAFSAs. Furthermore, we consider adapting  a utility-aware case filtering for unstructured event logs as a future research direction.

\section*{Acknowledgment} 

Work funded by European Research Council (PIX project) and by EU H2020-SU-ICT-03-2018 Project No.830929 CyberSec4Europe.


\appendix

\section{Proof of Privacy}
\label{apndx:proofs}
In this section, we prove that the proposed privacy mechanism (Algorithm 1) anonymizes an event log so that the guessing advantage that an attacker gains after the disclosing of the event log would not exceed $\delta$. Furthermore, we prove that the proposed method fulfills the requirements R1 and R2. As presented in the paper, the attacker has the following goals:

\begin{itemize}
    \item $h_1$: Has the case of an individual gone through a given prefix or suffix? The output is a bit with a value $\in \{0,1\}$.
    \item $h_2$: What is the cycle time of a particular activity that has been executed for the individual? The output is a real value that the attacker may wish to estimate with a certain precision.
\end{itemize}

\subsection{Guessing Advantage}
\label{apndx:guessing_advantage}
In this section, we prove that the adoption of the guessing advantage mechanism fulfills the requirement R2.

In this article, we apply the mechanism proposed by Laud et al.~\cite{laud2020framework,laud2019interpreting} to the release of the timestamp attribute of the event log.
In Def~\textbf{3.6}, we defined the guessing advantage as the difference between the posterior probability (after publishing $M(L)$) and the prior probability (before publishing $M(L)$) (prior knowledge) of an attacker making a successful guess in $H_p$.  
Def~\textbf{3.6} defines the prior knowledge.
Laud et al.~\cite{laud2020framework,laud2019interpreting} proposed estimation of the posterior guessing probability. 

\begin{prop}[Posterior Guessing Probability~\cite{laud2020framework,laud2019interpreting}]
\label{apndx:proposition_posterior}
	The posterior guessing probability of an attribute ranging between 0 and r for a single individual	 after the release of the timestamp attribute of an event log is bounded by	
	$P' \leq \frac{1}{1 + exp(-\epsilon \cdot r)\frac{1-P}{P}}\enspace.$	
\end{prop}

\begin{proof}
			(Taken from~\cite{laud2020framework,laud2019interpreting}) An attacker has a prior knowledge $k(l)$ of part of the event log $l$. Using the equality $Pr[X=x] = \sum_{y \in Y} Pr[X=x,Y=y]$ and Bayesian formula $Pr[A,B] = Pr[A|B]\cdot Pr[B]$, we can rewrite
	\begin{eqnarray*}\label{eq:prior}		
		P' &:=& Pr[h(L) \in H_p\ |\ M_f(L)=M_f(l), k(L) = k(l)]\\
		& = & \frac{Pr[h(L) \in H_p, M_f(L)=M_f(l), k(L) = k(l)]}{Pr[M_f(L)=M_f(l), k(L) = k(l)]}\\
		& = & \frac{\sum_{l': h(l') \in H_p, k(l) = k(l')}\ Pr[M_f(l')=M_f(l)]\cdot Pr[L = l']}{\sum_{l': k(l) = k(l')} Pr[M_f(l')=M_f(l)]\cdot Pr[L = l']}\\
		& = & \frac{1}{1 + \frac{\sum_{l':\ h(l') \notin h(l), k(l') = k(l)} Pr[M_f(l')=M_f(l)]\cdot Pr[L = l']}{\sum_{l'':\ h(l'') \in h(l), k(l'') = k(l)} Pr[M_f(l'')=M_f(l)]\cdot Pr[L = l'']}}\enspace,	
	\end{eqnarray*}
	For an $\epsilon$-DP mechanism $M_f$, since $l'$ and $l''$ differ in one item due to the condition $k(l')=k(l)=k(l'')$, we have $\frac{Pr[M_f(l')=M_f(l)]}{Pr[M_f(l'')=M_f(g)]} \geq exp(-\epsilon\cdot r)$, where $r$ is the largest possible difference between two values of an attribute that the attacker is guessing. This gives us	
	\begin{eqnarray*}
		P' &\leq& \frac{1}{1 + exp(-\epsilon\cdot r)\frac{\sum_{l':\ h(l') \notin h(l), k(l') = k(l)} Pr[L = l']}{\sum_{l'':\ h(l'') \in h(l), k(l'') = k(l)} Pr[L = l']}}\\
		&=& \frac{1}{1 + exp(-\epsilon\cdot r)\frac{Pr[h(L) \notin H_p, k(L) = k(l)]}{Pr[h(L) \in H_p, k(L) = k(l)]}}\\
		&=& \frac{1}{1 + exp(-\epsilon\cdot r)\frac{Pr[h(L) \notin H_p\ |\ k(L) = k(l)]}{Pr[h(L) \in H_p\ |\ k(L) = k(l)]}}\\
	\end{eqnarray*}
Substituting $P$ from Def~\textbf{3.6}, we get
\begin{equation}\label{eq:postpr}
	P' \leq \frac{1}{1 + exp(-\epsilon \cdot r)\frac{1-P}{P}}\enspace.
\end{equation}
\qedhere

\end{proof}

\begin{prop}\label{apndx:proposition_minimum_eps} 
The maximum possible $\epsilon$ that achieves the upper bound $\delta$, w.r.t the particular attack model, and fulfills the requirement $R2$ is
$\epsilon = \frac{-\ln\left(\frac{P}{1-P} \cdot \left(\frac{1}{\delta+P} - 1 \right) \right)}{r}$
\end{prop}
\begin{proof}
Given Def~\textbf{3.9}, $\delta$ is the maximum guessing advantage probability after disclosing the attribute. Therefore,
\begin{equation}\label{eq:guessing_adv_upper}
	P' - P \leq \delta \enspace.
\end{equation}
Substituting Eq~(\ref{eq:postpr}) into Eq~(\ref{eq:guessing_adv_upper}), we can estimate the largest possible $\epsilon$ (the minimum amount of noise) that achieves the upper bound $\delta$ as
\begin{equation}\label{eq:maineq}
\epsilon = \frac{-\ln\left(\frac{P}{1-P} \cdot \left(\frac{1}{\delta+P} - 1 \right) \right)}{r}\enspace.
\end{equation}

From Eq~(\ref{eq:maineq}), the $\epsilon$ represents the maximum possible value for $\epsilon$ achieving the bound $\delta$. Hence, we use Eq~(\ref{eq:maineq}) to get the minimum possible noise to fulfill the requirement $
R2$.\qedhere
\end{proof}

\subsection{Contingency Table (CT)}
\label{apndx:contingency_table}
In Def~\textbf{3.10}, we defined the DAFSA transition Contingency table as a histogram of counts of DAFSA transitions. First, we estimate the privacy leakage through the contingency table.
\begin{prop}\label{apndx:prop:ct:dp} 
Let $M^{CT}$ be a mechanism that, for each count cell, samples (independently) noise from the Laplace distribution $Lap(1/\epsilon)$ and adds it to the count. Then, the level of DP w.r.t. change in a prefix/suffix of some trace of the underlying event log is:
\begin{enumerate}
\item $\epsilon$ for a single count cell of the noisified CT;
\item $k\cdot\epsilon$ for the entire noisified CT, where $k$ is the longest case length in the event log.
\end{enumerate}
\end{prop}
\begin{proof}
Suppose that we update a prefix/suffix of a trace in the event log. 
Each single count cell in the contingency table may change by at most  $\pm 1$ for each update step.
Hence, the global sensitivity of a single count cell w.r.t. changing some prefix/suffix is $1$. We can sample additive noise from distribution $Lap(1/\epsilon)$ and add it to a count cell to achieve $\epsilon$-DP w.r.t. that count cell.

For the entire contingency table, a change in the frequency of a single case variant can affect at most $k$ count cells (i.e., the count over the DAFSA transitions that represent that case variant), where $k$ is the maximum case variant length in the event log. Hence, the global sensitivity is $k$, and we need to sample noise from $Lap(k/\epsilon)$ to achieve $\epsilon$-DP.\qedhere
\end{proof}


Given the desired guessing advantage threshold $\delta$, we need to prevent an attacker from achieving attack $h_1$.

\begin{prop}\label{apndx:prop:ct:ga} 
Let $M^{CT}$ be a mechanism that, for each count cell, draws (independently) noise from the Laplace distribution $Lap(1/\epsilon)$, adds or subtracts the noise from the count, and rounds negative results to 0. Let $\delta$ be the desired upper bound probability that the attacker achieves attack $h_1$. Then, it suffices to take
\[\epsilon=-\ln\left(\frac{P}{1-P} \cdot \left(\frac{1}{\delta+P} - 1 \right) \right)\enspace,\]
where $P = \frac{1-\delta}{2}$.
\end{prop}
\begin{proof}
Let $P$ be the prior probability of guessing (i.e., without observing the output). 
We adopt the mechanism proposed by Laud et al.~\cite{laud2020framework,laud2019interpreting} to estimate the guessing advantage, which is the upper bound of the difference between the posterior and the prior probability of guessing an attribute ranging between $0$ and $1$ (r=1) by
\[\delta = \frac{1}{1 + exp(-\epsilon)\frac{1-P}{P}} - P\enspace,\]
which can be reversed to
\[\epsilon = -\ln\left(\frac{P}{1-P} \cdot \left(\frac{1}{\delta+P} - 1 \right) \right)\enspace.\]

Laud et al.~\cite{laud2020framework,laud2019interpreting} estimate the prior knowledge $P$ from the distribution of input values. However, the contingency table contains counts, which means the lack of distribution. Laud et al.~\cite{laud2020framework,laud2019interpreting} elaborate that, in case of the lack of distribution, we can estimate the value of $P$ that minimizes the $\epsilon$. We take the derivative of $\epsilon$ w.r.t. $P$, which yields $P = \frac{1-\delta}{2}$ as the prior knowledge for the contingency table counts. 

We need to sample noise to achieve $\epsilon$-DP. First of all, rounding negative noisified counts up to $0$ can be considered as post-processing that does not depend on private data. 
Hence, we focus on adding the Laplace noise. While noise sampled from $Lap(1/\epsilon)$ would be enough for a single count cell, for an entire table, we would need $Lap(k/\epsilon)$ where $k$ is the longest case length in the event log, as shown in Prop~\ref{apndx:prop:ct:dp}. 

All correlated cells may provide additional information about the target cell.
For example, activity B may always follow activity A.
This additional knowledge may increase $P$. However, we have already chosen the \emph{worst-case} $P$ (in terms of guessing advantage) for estimating the noise, 
which does not depend on any background knowledge that the attacker may get, including the related outputs. We note that if we used 
another prior knowledge estimation and/or estimated the attacker's success directly instead of estimating the \emph{advantage}, then we would indeed require sampling noise from $Lap(k/\epsilon)$.\qedhere
\end{proof}

\subsection{Correctness of Sampling}
\label{apndx:sampling}

In this article, we adopt differential privacy, which injects noise into the data before its publication. 
In order to translate the calculated noise for every DAFSA transition group (in CT) into case variants anonymization, we perform case sampling, as defined in Def~\textbf{3.13} and Def~\textbf{3.14}. Following, we provide the correctness of sampling.
We assume that the event log contains the three columns: Case ID, Activity label, and Timestamps, and the activity labels are public information.
First, we discuss privacy leakage of the case variants distribution (\emph{without} taking into account timestamps).

\begin{prop}\label{apndx:prop:el:cnt:lap:ga}
Let the event log $L$ have a fixed constant timestamp value for all the entries. Let $CT$ be the contingency table that corresponds to the event log $L$. Let $M^{OV}$ be a mechanism that, for each count cell of the $CT$, draws (independently) noise $z$ from Laplace distribution $Lap(1/\epsilon)$ and:
\begin{itemize}

\item Replicates the random cases that go through the DAFSA transition, which corresponds to the count cell (prefix/suffix group), if $z >0$. The number of replications is $z$;
\item Drops random cases from the event log that go through a DAFSA transition which corresponds to the count cell, if $z < 0$. The number of deletions is $z$;
\item Shuffles all cases and updates all case IDs before publishing the resulting event log.

\end{itemize}
Let $\delta$ be the desired upper bound of the attacker succeeding in the goal $h_1$. Then, it suffices to take
\[\epsilon=-\ln\left(\frac{P}{1-P} \cdot \left(\frac{1}{\delta+P} - 1 \right) \right)\enspace,\]
where $P = \frac{1-\delta}{2}$.
\end{prop}
\begin{proof}
We show that sampling the cases has the same privacy guarantees as adding Laplace noise to the counts. Insertion and deletion of $z$ cases work similarly to adding Laplace noise with rounding negative results after noise injection up to $0$.
However, L contains strictly more information than the information represented by CT, e.g., the timestamp and the case IDs. 
An attacker can single out the individual using their timestamps (attack $h_2$). At this step, we assume having a fixed timestamp while anonymizing the case variants distribution, and we discuss the timestamp anonymization later in Appendix~\ref{appdx:time}.
Also, the attacker can guess the new injected cases through their case IDs.
For example, if the actual case IDs range from $0$ to $10$, and we start the injected new cases from case ID= $11$, the attacker will clearly see which cases are duplicates. In this article, we shuffle the cases and generate new case IDs for all the cases in the event log.
More formally, the final result is a multiset of cases, equivalent to the counts of a corresponding CT.\qedhere

\end{proof}

\subsection{Correctness of Oversampling}
In the conference version of this article, we fulfilled a stricter requirement. We fulfill the requirement that ``the anonymized log must have the same set of case variants as the original log.'' We adopted oversampling instead of sampling to fulfill such a requirement.
We use \emph{oversampling} to keep all the initial entries in the event log.

\begin{prop}\label{apndx:prop:el:cnt:onesidelap:ga}
Let the event log $L$ have a fixed timestamp value for all the entries. Let $CT$ be the contingency table that corresponds to the event log $L$. Let $M^{OV}$ be a mechanism that, for each count cell of the $CT$, samples (independently) noise $z$ from Laplace distribution $Lap(1/\epsilon)$ and:
\begin{itemize}
\item Inserts into $L$ $|z|$ additional copies of any cases that have been included in that count cell;
\item Shuffles all resulting cases and updates all case IDs before publishing the resulting event log.
\end{itemize}
Let $\delta$ be the desired upper bound of the attacker succeeding in the goal $h_1$. Then, 
\begin{itemize}
    \item it suffices to take 
    \[\epsilon=-2\ln(b / c^2 - (\delta - 1) / (c \cdot b))\enspace,\]
where $c = \sqrt[3]{6}$ and \\

$b = \sqrt[3]{\sqrt{3} \cdot \sqrt{2 \cdot \delta^3 + 21 \cdot \delta^2 - 48 \cdot \delta + 25} - 9 \cdot \delta + 9}$.
\item  a mechanism  that adds noise $|z|$, where $z \gets Lap(1/\epsilon)$, fulfills the requirement that ``the anonymized log must have the same set of case variants as the original log.''
\end{itemize}
\end{prop}
\begin{proof}
Oversampling the cases using Def~\textbf{3.12} keeps the same case variants as the input log L.
Similar to the proof of Prop~\ref{apndx:prop:el:cnt:lap:ga}, we show that a mechanism $M$ that adds noise $|z|$ for $z \gets Lap(1/\epsilon)$ to the counts of CT constructed from L ensures the differential privacy guarantees.
Let $L_1$ and $L_2$ be two neighboring event logs that differ in the presence of one trace. Without loss of generality, let $L_1$ have a smaller count than $L_2$ for the observed count cell. Suppose that the added noise instance is $|z| \geq 1$. In that case, the noisified output $y$ can be obtained from both $L_1$ and $L_2$. In particular, we have

\begin{eqnarray*}
\frac{\Pr[M(L_1) = y]}{\Pr[M(L_2) = y]} & = & \frac{exp(\epsilon \cdot (y - M(L_1)))}{exp(\epsilon \cdot (y - M(L_2)))}\\
& = & exp(\epsilon \cdot (y - M(L_1) + y - M(L_2)))\\
& \leq & exp(\epsilon)\enspace.
\end{eqnarray*}


So the guessing advantage $\delta$ for $|z| \geq 1$ can be computed as in Prop~\ref{apndx:prop:ct:ga}. However, for $|z| < 1$, the noisified value $y$ produced by $M(L_1)$ will be \emph{between} the true counts of $L_1$ and $L_2$, and $y$ could never be an output of $M(L_2)$, since we only add positive noise. 
The probability of getting $|z| \geq 1$ for Laplace noise is $CDF_{LAP}(1) - CDF_{LAP}(-1) = (1 - \frac{1}{2}exp(-\epsilon)) -  \frac{1}{2}exp(-\epsilon) = 1 - exp(-\epsilon)$. 

Let $\epsilon$ is computed to achieve guessing advantage $\delta'$ with standard Laplace distribution as in Prop~\ref{apndx:prop:el:cnt:lap:ga}. The actual guessing advantage with one-sided Laplace distribution will be
\begin{equation}\label{it:1}
\delta = exp(-\epsilon) \cdot \delta' + (1 - exp(-\epsilon))\enspace.
\end{equation}

From Prop~\ref{apndx:proposition_minimum_eps}, the maximum $\epsilon$ that achieves the upper bound $\delta$ is $\epsilon=-\ln\left(\frac{P}{1-P} \cdot \left(\frac{1}{\delta'+P} - 1 \right) \right)$ for $P = \frac{1-\delta'}{2}$.
Substituting from (\ref{eq:maineq}) into (\ref{it:1}) and using a numerical solver (e.g., Wolfram Alpha) to express $\epsilon$ through $\delta$.

 \[\epsilon=-2\ln(b / c^2 - (\delta - 1) / (c \cdot b))\enspace,\]
where $c = \sqrt[3]{6}$ and \\

$b = \sqrt[3]{\sqrt{3} \cdot \sqrt{2 \cdot \delta^3 + 21 \cdot \delta^2 - 48 \cdot \delta + 25} - 9 \cdot \delta + 9}$.


\qedhere 
\end{proof}

\subsection{Timestamp Anonymization}
\label{appdx:time}
In this appendix, we consider timestamp anonymization. We need to add enough noise to the timestamps to prevent attacker success in $h_2$, and $h_1$, since timestamps may potentially leak which traces are real and which are not (new injected cases by sampling). 
We note that correlations between timestamps are difficult to handle in terms of guessing advantage. The main problem is that the effect of correlated times would greatly depend on how they are correlated, and in some cases, even a little leakage of $ts_2$ may leak everything about $ts_1$. For example, let $ts_1 \in \{0,1\}$ and $ts_2 \in \{0,...,1023\}$ be distributed uniformly. We have the prior knowledge $P_1 = 1/2$, and $P_2 = 1/1024$. Suppose that, after seeing the noisified output, the attacker constrained his view of $ts_2$ to $\{0,...,511\}$. There are now $512$ possible choices for $ts_2$, so $\delta = 1/512 - 1/1024 = 1/1024$, which is very small. However, when $ts_1$ is the highest bit of $ts_2$, the value of $ts_1$ would be leaked. This particular correlation of times is unlikely in practice, but it demonstrates the problem in general.

To this end, we convert a timestamp $ts_k$ to a time difference $dts_k := ts_k - ts_{k-1}$ of sequential events, which is the duration of time that an individual has spent in a transition from one event to the next one. The timestamp at time $ts_0$ is the starting time $st$, which needs to be published as well to make it possible to reconstruct the actual timestamps.  
$st$ can be converted to the time difference between the start time of a case and the start time of the log.
We obtain the result for linearly correlated time differences (e.g., increasing the duration of the first event by a minute increases every other event at most by one minute). To this end, we scale $\epsilon_k$ by the length of the trace. Prop~\ref{apndx:prop:el:time:ga} proposes a mechanism that anonymizes the time differences. 

\begin{prop}\label{apndx:prop:el:time:ga}
Let $M^{T}$ be a mechanism that, for each timestamp $ts_k$, samples (independently) noise from the Laplace distribution $Lap(1/\epsilon)$ and adds the noise to $ts_k$. Let $r_k$ the maximum possible value of $ts_k$. Let $\delta$ be the desired upper bound of the attacker succeeding in the attack $h_2$ with precision $p$. Then, it suffices to take
\[\epsilon_k=-\frac{\ln\left(\frac{P_k}{1-P_k} \cdot \left(\frac{1}{\delta+P_k} - 1 \right) \right)}{m \cdot r_k}\enspace,\]
where $P_k = CDF(ts_k + p \cdot r_k) - CDF(ts_k - p \cdot r_k)$, $CDF$ is the probability density function of the distribution of times, and $m$ is the length of the longest trace in L.
\end{prop}
\begin{proof}
Similarly to Prop~\ref{apndx:prop:el:cnt:lap:ga}, we can estimate an upper bound on the difference between the posterior and the prior probability of guessing an attribute ranging between $0$ and $1$ by
\[\delta = \frac{1}{1 + exp(-\epsilon_k \cdot r_k)\frac{1-P_k}{P_k}} - P_k\enspace,\]
which can be reversed to
\[\epsilon_k = -\frac{\ln\left(\frac{P_k}{1-P_k} \cdot \left(\frac{1}{\delta+P_k} - 1 \right) \right)}{r_k}\enspace.\]
The quantity $P_k$ is defined as the probability of a value being in the interval $[ts_k - p\cdot r_k, ts_k + p\cdot r_k]$.
From Prop~\ref{apndx:proposition_minimum_eps}, the above value of $\epsilon$ is the maximum within the upper bound $\delta$.
Similarly to Prop~\ref{apndx:prop:el:cnt:lap:ga}, if $CDF$ of time distributions is unknown, it is safe to take $P_k$ that minimizes the $\epsilon_k$ (and hence maximizes the amount of noise), which is $P_k = \frac{1-\delta}{2}$ for all $k$.

That was about leakage for a single published timestamp. We could keep $\epsilon_k$ the same for an entire L if different times are not correlated. 
If the times are linearly correlated, then, by the sequential composition of DP, we need to divide $\epsilon_k$ by the length of the longest trace $m$. \qedhere 
\end{proof}

\subsection{Event Log Anonymization}

We now provide a proof of correctness of Alg.~\textbf{1}.

\begin{lem}
\label{lem:dpel}
Let $L$ and $\delta$ be an event log and a maximum guessing advantage threshold. Log $L'$ = Alg1($L$,$\delta$) fulfills the two properties in Def~\textbf{3.1} (i.e.,\ 
$L'$ is $\Phi$-personalized differentially private), and $L'$ fulfills requirements~\textbf{R1} and~\textbf{R2}.
\end{lem}
\begin{proof}
We first show that $L'$ fulfills the two properties in Def~\textbf{3.1}.

\begin{enumerate}[label={(\arabic*)}]
\item  In Prop.~\ref{apndx:prop:el:cnt:lap:ga}, we proved that it suffices to take 
\[\epsilon_d=-\ln\left(\frac{P}{1-P} \cdot \left(\frac{1}{\delta+P} - 1 \right) \right) \enspace\] 
to anonymize the case variants of an event log in order to provide differential privacy guarantees for the desired guessing advantage upper bound $\delta$. Algorithm~\textbf{1} annotates the event log using DAFSA transitions (line 3), and computes the contingency table $CT$ of the log. After that, the algorithm uses the above $\epsilon_d$ to sample cases (replication and deletion of cases) (lines 18 and 21). Given that Prop.~\ref{apndx:prop:ct:ga} and Prop.~\ref{apndx:prop:el:cnt:lap:ga} prove that replicating and deleting cases with a random sample size $z$ from Laplace distribution $Lap(1/\epsilon)$ provides differential privacy guarantees. Algorithm~\textbf{1} provides differential privacy guarantees for the case variants anonymization;
\item In Prop.~\ref{apndx:prop:el:time:ga}, we proved that it suffices to take 
\[\epsilon_k = -\frac{\ln\left(\frac{P_k}{1-P_k} \cdot \left(\frac{1}{\delta+P_k} - 1 \right) \right)}{r_k} \enspace\] 
to provide differential privacy guarantees w.r.t timestamp for the desired guessing advantage upper bound $\delta$. Given that Algorithm~\textbf{1} adopts the above $\epsilon$ to anonymize timestamp each event ($\Phi_t$)(line 27). Hence, Algorithm~\textbf{1} provides differential privacy guarantees for the event log anonymization w.r.t. timestamp.
\end{enumerate}

We now show that $L'$ fulfills the two requirements~\textbf{R1} and~\textbf{R2}.

\begin{enumerate}
\item Algorithm~\textbf{1} does not create new case variants. It only replicates or deletes existing variants (cf.\ lines 18 and 21). Hence, it fulfills requirement~\textbf{R1};
\item In Prop.~\ref{apndx:proposition_minimum_eps}, we proved that the $\epsilon_k$ as defined in Eq~(\ref{eq:maineq}) is the maximum within the upper bound $\delta$, so Algorithm~\textbf{1} fulfills the requirement~\textbf{R2}.
\end{enumerate}

\end{proof}

\section{Evaluation Supplementary Material}
\label{sec:eval}



\subsection{Dataset Selection}
\begin{table*}

	\centering
	\caption{Descriptive Statistics of Event Logs}
	\label{tab:event_logs}
	\footnotesize
	\begin{tabular}{|c|c|c|c|c|c|c|c|c|c|c|c|}
		\hline
		
		\multirow{2}{*}{event log}
		&\multirow{2}{2em}{\centering\# Traces}
		&\multirow{2}{1.5em}{\centering\# Tasks}
		&\multirow{2}{2em}{\centering\# Events}
		&\multirow{2}{3em}{\centering\# Edges}
		&\multirow{2}{3em}{\centering Case Variant}
		&\multicolumn{2}{c|}{Trace Length}
		& \multicolumn{3}{c|}{Case Duration}\\
		
		\cline{7-11}
		
		& & & & & &  Min&Max&Min &Max &Avg\\

		\hline
		
		$BPI12$~\cite{BPIC2012}	&	13087	&	23	&	262200	&	116	&	4366	&	3	&	175	&	1.85 s	&	4.51 m	&	1.23 w	\\
				\hline
		$BPI13_{i}$~\cite{BPIC13}	&	7554	&	4	&	65533	&	16	&	1511	&	1	&	123	&	inst.	&	2.11 y	&	1.73 w	\\
				\hline
		$BPI14_{i}$~\cite{BPIC14}	&	46616	&	39	&	466737	&	497	&		22632	&	1	&	178	&	14 s	&	1.07 y	&	5.07 d	\\
				\hline
		$BPI15_1$~\cite{BPIC15}	&	1199	&	398	&	52217	&	495	&	1170	&	2	&	101	&	8.56 h	&	4.07 y	&	3.15 m	\\
				\hline
		$BPI17$~\cite{BPIC2017}	&	31509	&	24	&	1202267	&	181		&	3942	&	10	&	180	&	3.35 m	&	9.4 m	&	3.13 w	\\
				\hline
		$BPI18$~\cite{BPIC18}	&	43809	&	14	&	2514266	&	499		&	28457	&	24	&	2973	&	3.74 m	&	2.77 y	&	11.03 m	\\
		        \hline
		
		$BPIC19$~\cite{BPIC2019}	&	251734	&	42	&	1595923	&	498		&	11973	&	1	&	990	&	2 ms	&	70.33 y	&	2.35 m	\\
				\hline
		$BPI20_{r}$~\cite{BPIC20}	&	7065	&	51	&	86581	&	500	&		1478	&	3	&	90	&	12.61 h	&	3.26 y	&	2.87 m	\\
				\hline
		$CCC19$~\cite{CCC2019}	&	20	&	29	&	1394	&	149	&	20	&	52	&	118	&	11 m	&	1.01 d	&	1.73 h	\\
				\hline
		$CredReq$~\cite{creditReq}	&	10035	&	8	&	150525	&	9		&	1	&	15	&	15	&	3.5 h	&	5 d	&	22 h	\\
				\hline
		$Hospital$~\cite{hospital}	&	1143	&	624	&	150291	&	903		&	981	&	1	&	1814	&	inst.	&	3.17 y	&	1.06 y	\\
				\hline
		$Sepsis$~\cite{mannhardt2016sepsis}	&	1050	&	16	&	15214	&	115		&	846	&	3	&	185	&	2.03 m	&	1 y	&	4 w	\\
				\hline
		$Traffic$~\cite{traffic}	&	150370	&	11	&	561470	&	77		&	231	&	2	&	20	&	3 d	&	12 y	&	11 m	\\
				\hline
		$Unrine.$~\cite{Gunst2020}	&	1650	&	10	&	6973	&	25		&	50	&	2	&	35	&	10.1 m	&	2.32 y	&	3.7 w	\\
				\hline

	\end{tabular}
\end{table*}

Table~\ref{tab:event_logs} shows the selected event logs and their descriptive statistics. 
The set of selected logs (cf. Table~\ref{tab:event_logs}) can be divided into two categories: structured and unstructured logs. Structured logs have a defined process, and most of the cases follow similar case variants, e.g., Unrineweginfectie, Credit Requirement, and Road Traffic Fines. On the other hand, unstructured logs have a high degree of uniqueness, which appears in more case variants of the log, e.g., Sepsis cases, CCC19, Hospital, BPI challenge 12, 13, 14, 15, 18, and 19.


\subsection{Folder Organization}
All the anonymized logs from the experiments are available in~\cite{anonymized_logs}.
The main folder is ``anonymized logs''. There are 3 subdirectories: Amun which contains the anonymized logs by the proposed approach with different settings, PRIPEL which contains the anonymized logs by PRIPEL framework, and SaCoFa which contains the anonymized logs by the SaCoFa framework.
For the event logs anonymized by Amun, the event attribute ``epsilon\_per\_event'' represents the $\epsilon$ used to anonymize the time component of this event. The case attribute ``epsilon\_per\_trace'' represents the $\epsilon$ value used to anonymize the frequency of all the instances of this case variant.

\subsection{EMD for the minimum and the maximum $\Phi$ value}

Tables~\ref{tbl:emd_min} and~\ref{tbl:emd_max}  present the empirical evaluation with the min($\Phi$) and the max($\Phi$)as the input $\epsilon$ to the state-of-the-art, respectively.

\begin{table*} [hbtp]

	\caption{Earth Movers' Distance for the output of different anonymization approaches using the minimum value of $\Phi$. A ``-'' means that the approach ran out of memory or timed out.}
	\label{tbl:emd_min}
\scriptsize

\centering
	\begin{tabular}{ p{0.9cm} p{0.3cm} p{0.6cm} c c c c cc c c c  }
		\toprule
        \multirow{2}{*}{Log}		&	\multirow{2}{*}{$\delta$} 	&	\multirow{2}{*}{min($\Phi$)}	&	 \multicolumn{4}{c}{EMD Freq} &	 \multicolumn{4}{c}{EMD Time}	\\
       
        	\cmidrule(lr){4-8} \cmidrule(l){9-12}	&		&		&	   $Amun_s$	&	  $Amun_f$	&	  $Amun_o$	&	  $PRIPEL$ & $SaCoFa$	&	  $Amun_s$	&	  $Amun_f$	&	 $Amun_o$	&	 $PRIPEL$	\\
        \hlineB{2}
        
        \multirow{3}{*}{BPIC12}	&	0.2	&	0.058	&	\textbf{331.02}	&	653.36	&	2301.62	&	-	&936.55 &	40.30	&	\textbf{8.08}	&	192.13	&	-	\\
        	&	0.3	&	0.088	&	\textbf{212.64}	&	742.39	&	1597.79	&	-	&	898.55& 20.32	&	\textbf{13.54}	&	107.96	&	-	\\
        	&	0.4	&	0.088	&	\textbf{142.37}	&	785.87	&	1275.43	&	-	& 909.64&	\textbf{12.68}	&	16.74	&	72.56	&	-	\\\hline

        \multirow{3}{*}{BPIC13}	&	0.2	&	0.07	&	1131.91	&	\textbf{1053.45}	&	7613.27	&	-	& 3589.36&	778.18	&	811.51	&	3343.92	&	-	\\
        	&	0.3	&	0.1	&	840.55	&	\textbf{258.45}	&	5417.64	&- & 3777.81	&592.36	&	307.26	&	2055.23	&	-	\\
        	&	0.4	&	0.1	&	\textbf{558.45}	&	2450.45	&	4296.73	&	-	& 3795.81&	486.86	&	\textbf{75.13}	&	1751.96	&	-	\\\hline

        \multirow{3}{*}{BPIC14}	&	0.2	&	0.034	&	429.60	&	\textbf{395.24}	&	1905.69	&	-	& 473.48&	132.53	&	130.00	&	577.32	&	-	\\
        	&	0.3	&	0.07	&	298.86	&	\textbf{281.43}	&	1333.95	&	-	& 456.5&	82.02	&	\textbf{76.37}	&	321.72	&	-	\\
        	&	0.4	&	0.1	&	208.55	&	\textbf{179.94}	&	1012.62	&	-	& 489.47&	\textbf{47.12}	&	50.61	&	178.08	&	-	\\\hline

        \multirow{3}{*}{BPIC15}	&	0.2	&	0.043	&	20.71	&	\textbf{18.74}	&	80.61	&	-	&20.765&	5.68	&	\textbf{4.18}	&	22.47	&	-	\\
        	&	0.3	&	0.07	&	14.82	&	\textbf{7.93}	&	52.42	&	-	&	10.8&3.15	&	2.28	&	11.17	&	-	\\
        	&	0.4	&	0.88	&	12.60	&	\textbf{2.79}	&	39.09	&	10.84	&10.832&	2.46	&	\textbf{1.05}	&	8.87	&	585.20	\\\hline

        \multirow{3}{*}{BPIC17}	&	0.2	&	0.058	&	\textbf{141.37}	&	1925.92	&	1159.08	&	-	&1488&	117.56	&	\textbf{75.98}	&	268.08	&	-	\\
        	&	0.3	&	0.07	&	\textbf{78.60}	&	2667.91	&	938.79	&	-	&1515.90&	\textbf{95.07}	&	131.76	&	206.76	&	-	\\
        	&	0.4	&	0.1	&	\textbf{49.93}	&	2674.76	&	938.79	&	-	&1321.6&	\textbf{79.33}	&	133.77	&		206.76&	-	\\\hline

        \multirow{3}{*}{BPIC18}	&	0.2	&	0.039	&	3775.468&	\textbf{659.01}&	17668.46
					&	-	&3931.28&	2176.17 &	\textbf{179.39}	& 3206.75
					&	-	\\
        	&	0.3	&	0.058	&	2922.80	&	\textbf{1752.72}	&	12063.81	&	-	&3600.65&	1319.69	&	\textbf{325.14}	&	2201.39	&	-	\\
        	&	0.4	&	0.054	&	\textbf{2887.08}	&	2812.65	& 9055.04		&	-	&3028.82&	885.81	&	\textbf{517.59}	&	1647.15	&	-	\\\hline

        \multirow{3}{*}{BPIC19}	&	0.2	&	0.045	&	946.99	&	\textbf{811.80}	&	4509.08	&	-	&1163.68&	\textbf{524.18}	&	608.02	&	1513.92	&	-	\\
        	&	0.3	&	0.078	&	743.96	&	\textbf{572.71}	&	3094.05	&	-	&	1141.659&\textbf{491.84}	&	498.59	&	1054.23	&	-	\\
        	&	0.4	&	0.078	&	612.48	&	\textbf{399.36}	&	2376.21	&	-	&1024.63&	500.57	&	\textbf{360.54}	&	751.67	&	-	\\\hline

        \multirow{3}{*}{BPIC20}	&	0.2	&	0.07	&	18.97	&	\textbf{18.69}	&	138.05	&	-	&84.58&	65.23	&	69.04	&	180.50	&	-	\\
        	&	0.3	&	0.14	&	14.88	&	\textbf{10.42}	&	99.76	&	-	&84.004&	44.46	&	\textbf{40.01}	&	109.90	&-	\\
        	&	0.4	&	0.233	&	10.55	&	\textbf{2.24}	&	76.61	&	-	&86.27&	41.70	&	\textbf{35.78}	&	82.67	&	-	\\ \hline

        \multirow{3}{*}{CCC19}	&	0.2	&	0.07	&	12.78	&	\textbf{3.70}	&	30.81	&	-	&-&	\textbf{0.00}	&	\textbf{0.00}	&	\textbf{0.00}	&	-	\\
        	&	0.3	&	0.14	&	4.55	&	\textbf{2.23}	&	25.77	&	-	&	-&\textbf{0.00}	&	\textbf{0.00}	&	\textbf{0.00}	&	-	\\
        	&	0.4	&	0.233	&	5.54	&	\textbf{3.21}	&	16.64	&	-	&-&	\textbf{0.00}	&	\textbf{0.00}	&	\textbf{0.00}	&	-	\\\hline

        \multirow{3}{*}{CredReq}	&	0.2	&	0.7	&	\textbf{0.00}	&	2.00	&	4.00	&	\textbf{0.00}	&\textbf{0.00}&	233.08	&	238.79	&	186.94	&	\textbf{0.00}	\\
        	&	0.3	&	0.7	&	\textbf{0.00}	&	\textbf{0.00}	&	3.00	&	\textbf{0.00}	&\textbf{0.00}&	203.94	&	182.39	&	220.17	&	\textbf{0.00}	\\
        	&	0.4	&	0.7	&	\textbf{0.00}	&	2.00	&	2.00	&	\textbf{0.00}	&\textbf{0.00}&	201.81	&	231.37	&	180.77	&	\textbf{0.00}	\\ \hline

        \multirow{3}{*}{Hospital}	&	0.2	&	0.048	&	81.40	&	\textbf{75.74}	&	406.45	&	-	&85.19&	\textbf{10.30}	&	11.33	&	55.44	&	-	\\
        	&	0.3	&	0.088	&	\textbf{57.91}	&	61.96	&	264.09	&	-	&85.25&	7.86	&	\textbf{7.42}	&	31.71	&	-	\\
        	&	0.4	&	0.14	&	\textbf{58.68}	&	63.93	&	210.40	&	-	&85.23&	6.07	&	\textbf{6.01}	&	20.70	&	-	\\ \hline

        \multirow{3}{*}{Sepsis}	&	0.2	&	0.07	&	56.84	&	\textbf{32.20}	&	427.64	&	-	&118.91&	8.97	&	\textbf{4.37}	&	61.50	&	-	\\
        	&	0.3	&	0.088	&	\textbf{28.46}	&	67.97	&	286.23	&	-	&121.42&	6.35	&	\textbf{3.60}	&	32.93	&	-	\\
        	&	0.4	&	0.14	&	\textbf{43.38}	&	101.64	&	232.50	&	-	&120.91&	\textbf{4.26}	&	6.48	&	26.61	&	-	\\ \hline

        \multirow{3}{*}{Traffic}	&	0.2	&	0.203	&	1.64	&	\textbf{0.90}	&	33.50	&	-	&132.13&	8250.42	&	\textbf{7253.29}	&	8627.86	&	-	\\
        	&	0.3	&	0.619	&	\textbf{0.61}	&	8.01	&	25.51	&	-	&77.33&	7767.32	&	\textbf{6720.80}	&	7081.06	&	-	\\
        	&	0.4	&	0.7	&	\textbf{0.00}	&	2951.50	&	21.00	&	-	&21.34&	\textbf{7182.96}	&	12788.57	&	7419.02	&	-	\\ \hline

        \multirow{3}{*}{Unrine.}	&	0.2	&	0.203	&	6.53	&	\textbf{5.53}	&	66.00	&	289.4	&111.80&	44.96	&	\textbf{39.34}	&	184.68	&	68087	\\
        	&	0.3	&	1.238	&	\textbf{0.00}	&	2.87	&	53.67	&	287	&113.20&	\textbf{28.12}	&	33.27	&	103.67	&	67949.3	\\
        	&	0.4	&	1.417	&	\textbf{1.47}	&	2.53	&	47.00	&	291.60	&113.80&	\textbf{21.89}	&	26.35	&	90.93	&	68233.90	\\ 
\bottomrule

	\end{tabular}
	
\end{table*}

\begin{table*}[hbtp]

	\caption{Earth Movers' Distance for the output of different anonymization approaches using the maximum value of $\Phi$. A ``-'' means that the approach ran out of memory or timed out.}
	\label{tbl:emd_max}
\scriptsize

\centering
	\begin{tabular}{ p{0.9cm} p{0.3cm} p{0.6cm} c c c c cc c c c  }
		\toprule
        \multirow{2}{*}{Log}		&	\multirow{2}{*}{$\delta$} 	&	\multirow{2}{*}{max($\Phi$)}	&	 \multicolumn{4}{c}{EMD Freq} &	 \multicolumn{4}{c}{EMD Time}	\\
       
        	\cmidrule(lr){4-8} \cmidrule(l){9-12}	&		&		&	   $Amun_s$	&	  $Amun_f$	&	  $Amun_o$	&	  $PRIPEL$ & $SaCoFa$	&	  $Amun_s$	&	  $Amun_f$	&	 $Amun_o$	&	 $PRIPEL$	\\
        \hlineB{2}
        
        \multirow{3}{*}{BPIC12}	&	0.2	&	3.09	&	\textbf{331.02}	&	653.36	&	2301.62	&	1006.80	&1006.62 &	40.30	&	\textbf{8.08}	&	192.13	&	19317	\\
        	&	0.3	&	3.33	&	\textbf{212.64}	&	742.39	&	1597.79	&	977.42	&	855.57& 20.32	&	\textbf{13.54}	&	107.96	&	19264.25	\\
        	&	0.4	&	9.97	&	\textbf{142.37}	&	785.87	&	1275.43	&	905.11	& 555.66&	\textbf{12.68}	&	16.74	&	72.56	&	18101.72	\\\hline

        \multirow{3}{*}{BPIC13}	&	0.2	&	3.13	&	1131.91	&	\textbf{1053.45}	&	7613.27	&	3713.82	& 3874.63&	778.18	&	811.51	&	3343.92	&	140484	\\
        	&	0.3	&	3.24	&	840.55	&	\textbf{258.45}	&	5417.64	&	3056.55 & 2650.64	&592.36	&	307.26	&	2055.23	&	112652.4	\\
        	&	0.4	&	10	&	\textbf{558.45}	&	2450.45	&	4296.73	&	2225.18	& 3024.18&	486.86	&	\textbf{75.13}	&	1751.96	&	73357.09	\\\hline

        \multirow{3}{*}{BPIC14}	&	0.2	&	2.92	&	429.60	&	\textbf{395.24}	&	1905.69	&	535.4	& 526.49&	132.53	&	130.00	&	577.32	&	7217.31	\\
        	&	0.3	&	3.21	&	298.86	&	\textbf{281.43}	&	1333.95	&	531.18	& 535.05&	82.02	&	\textbf{76.37}	&	321.72	&	7199.76	\\
        	&	0.4	&	3.6	&	208.55	&	\textbf{179.94}	&	1012.62	&	508.15	& 507.28&	\textbf{47.12}	&	50.61	&	178.08	&	7143.46	\\\hline

        \multirow{3}{*}{BPIC15}	&	0.2	&	2.65	&	20.71	&	\textbf{18.74}	&	80.61	&	20.77	&20.85&	5.68	&	\textbf{4.18}	&	22.47	&	579.44	\\
        	&	0.3	&	2.79	&	14.82	&	\textbf{7.93}	&	52.42	&	10.77	&	10.85&3.15	&	2.28	&	11.17	&	583.81	\\
        	&	0.4	&	2.9	&	12.60	&	\textbf{2.79}	&	39.09	&	10.84	&10.42&	2.46	&	\textbf{1.05}	&	8.87	&	583.8	\\\hline

        \multirow{3}{*}{BPIC17}	&	0.2	&	2.17	&	\textbf{141.37}	&	1925.92	&	1159.08	&	2776.74	&2385.81&	117.56	&	\textbf{75.98}	&	268.08	&	78550.82	\\
        	&	0.3	&	3.38	&	\textbf{78.60}	&	2667.91	&	938.79	&	2776.7	&2471.76&	\textbf{95.07}	&	131.76	&	206.76	&	102194	\\
        	&	0.4	&	3.61	&	\textbf{49.93}	&	2674.76	&	938.79	&	1946.15	&1280.19&	\textbf{79.33}	&	133.77	&		206.76&	61156.46	\\\hline

        \multirow{3}{*}{BPIC18}	&	0.2	&	2.82	&	3775.468&	\textbf{659.01}&	17668.46
					&	-	&2587.28&	2176.17 &	\textbf{179.39}	& 3206.75
					&	-	\\
        	&	0.3	&	2.99	&	2922.80	&	\textbf{1752.72}	&	12063.81	&	-	&2494.17&	1319.69	&	\textbf{325.14}	&	2201.39	&	-	\\
        	&	0.4	&	3.28	&	\textbf{2372.08}	&	2812.65	& 9055.04		&	3537.58	&2372.69&	885.81	&	\textbf{517.59}	&	1647.15	&	504500.1	\\\hline

        \multirow{3}{*}{BPIC19}	&	0.2	&	4.30	&	946.99	&	\textbf{811.80}	&	4509.08	&	-	&872.75&	\textbf{524.18}	&	608.02	&	1513.92	&	-	\\
        	&	0.3	&	4.77	&	743.96	&	\textbf{572.71}	&	3094.05	&	-	&	871.43&\textbf{491.84}	&	498.59	&	1054.23	&	-	\\
        	&	0.4	&	5.17	&	612.48	&	\textbf{399.36}	&	2376.21	&	-	&813.17&	500.57	&	\textbf{360.54}	&	751.67	&	-	\\\hline

        \multirow{3}{*}{BPIC20}	&	0.2	&	3.79	&	18.97	&	\textbf{18.69}	&	138.05	&	98.61	&142.92&	65.23	&	69.04	&	180.50	&	17499.66	\\
        	&	0.3	&	4.11	&	14.88	&	\textbf{10.42}	&	99.76	&	73.03	&128.32&	44.46	&	\textbf{40.01}	&	109.90	&11002.3	\\
        	&	0.4	&	4.37	&	10.55	&	\textbf{2.24}	&	76.61	&	52.15	&126.68&	41.70	&	\textbf{35.78}	&	82.67	&	8141.17	\\ \hline

        \multirow{3}{*}{CCC19}	&	0.2	&	1.85	&	12.78	&	\textbf{3.70}	&	30.81	&	-	&7.17&	\textbf{0.00}	&	\textbf{0.00}	&	\textbf{0.00}	&	-	\\
        	&	0.3	&	2.67	&	4.55	&	\textbf{2.23}	&	25.77	&	-	&	4.94&\textbf{0.00}	&	\textbf{0.00}	&	\textbf{0.00}	&	-	\\
        	&	0.4	&	2.99	&	5.54	&	\textbf{3.21}	&	16.64	&	-	&3.88&	\textbf{0.00}	&	\textbf{0.00}	&	\textbf{0.00}	&	-	\\\hline

        \multirow{3}{*}{CredReq}	&	0.2	&	7.8	&	\textbf{0.00}	&	2.00	&	4.00	&	\textbf{0.00}	&0&	233.08	&	238.79	&	186.94	&	\textbf{0.00}	\\
        	&	0.3	&	8.36	&	\textbf{0.00}	&	\textbf{0.00}	&	3.00	&	\textbf{0.00}	&0&	203.94	&	182.39	&	220.17	&	\textbf{0.00}	\\
        	&	0.4	&	8.8	&	\textbf{0.00}	&	2.00	&	2.00	&	\textbf{0.00}	&0&	201.81	&	231.37	&	180.77	&	\textbf{0.00}	\\ \hline

        \multirow{3}{*}{Hospital}	&	0.2	&	2.60	&	81.40	&	\textbf{75.74}	&	406.45	&	85.22	&85.24&	\textbf{10.30}	&	11.33	&	55.44	&	2469.24	\\
        	&	0.3	&	3.29	&	\textbf{57.91}	&	61.96	&	264.09	&	85.22	&83.87&	7.86	&	\textbf{7.42}	&	31.71	&	2469.24	\\
        	&	0.4	&	3.55	&	\textbf{58.68}	&	63.93	&	210.40	&	85.19	&82.71&	6.07	&	\textbf{6.01}	&	20.70	&	2381.80	\\ \hline

        \multirow{3}{*}{Sepsis}	&	0.2	&	3.02	&	56.84	&	\textbf{32.20}	&	427.64	&	119.68	&121.22&	8.97	&	\textbf{4.37}	&	61.50	&	6218.25	\\
        	&	0.3	&	3.29	&	\textbf{28.46}	&	67.97	&	286.23	&	118.71	&102.39&	6.35	&	\textbf{3.60}	&	32.93	&	6217.93	\\
        	&	0.4	&	3.55	&	\textbf{43.38}	&	101.64	&	232.50	&	118.71	&97.80&	\textbf{4.26}	&	6.48	&	26.61	&	6217.93	\\ \hline

        \multirow{3}{*}{Traffic}	&	0.2	&	4.93	&	1.64	&	\textbf{0.90}	&	33.50	&	-	&135.26&	8250.42	&	\textbf{7253.29}	&	8627.86	&	-	\\
        	&	0.3	&	5.49	&	\textbf{0.61}	&	8.01	&	25.51	&	-	&121.62&	7767.32	&	\textbf{6720.80}	&	7081.06	&	-	\\
        	&	0.4	&	5.74	&	\textbf{0.00}	&	2951.50	&	21.00	&	-	&75.87&	\textbf{7182.96}	&	12788.57	&	7419.02	&	-	\\ \hline

        \multirow{3}{*}{Unrine.}	&	0.2	&	3.77	&	6.53	&	\textbf{5.53}	&	66.00	&	113.87	&113.53&	44.96	&	\textbf{39.34}	&	184.68	&	58296	\\
        	&	0.3	&	4.31	&	\textbf{0.00}	&	2.87	&	53.67	&	49.46	&94.4&	\textbf{28.12}	&	33.27	&	103.67	&	30038.10	\\
        	&	0.4	&	4.42	&	\textbf{1.47}	&	2.53	&	47.00	&	44.33	&42.8&	\textbf{21.89}	&	26.35	&	90.93	&	27204.87	\\ 
\bottomrule

	\end{tabular}
	
\end{table*}


 \bibliographystyle{elsarticle-num} 
 \bibliography{Amun_ext}





\end{document}